\documentclass[journal, 10pt]{IEEEtran}
\usepackage{cite}
\usepackage{epsfig}
\usepackage{times}
\usepackage{graphicx}
\usepackage[figuresright]{rotating}
\usepackage{pict2e}
\usepackage{algorithmic}
\usepackage{algorithm}
\usepackage{amssymb}
\usepackage{amsmath}
\usepackage{setspace}
\usepackage{comment}
\usepackage{balance}
\usepackage{rotating}
\usepackage{multirow}
\usepackage{wrapfig}
\usepackage{mathtools}
\usepackage{booktabs}
\usepackage{gensymb}
\usepackage{array}
\usepackage[labelformat=simple]{subcaption}

\usepackage{colortbl}
\newcommand{\todo}[1]
{{\color{red}\textbf{#1}}}
\newcommand{\todob}[1]
{{\color{blue}\textbf{#1}}}

\usepackage{amssymb,amsmath,amsthm}

\newtheorem{proposition}{Proposition}

\newtheorem{remark}{Remark}
\renewcommand{\theequation}{\arabic{equation}}
\usepackage{bm}

\def\BibTeX{{\rm B\kern-.05em{\sc i\kern-.025em b}\kern-.08em
    T\kern-.1667em\lower.7ex\hbox{E}\kern-.125emX}}
    
\begin{document}
\title{Location-aware Beamforming for MIMO-enabled UAV Communications: An Unknown Input Observer Approach} %

\author{\IEEEauthorblockN{ 
Alireza Mohammadi\IEEEauthorrefmark{1}, 
Mehdi Rahmati\IEEEauthorrefmark{2}, and 
Hafiz Malik\IEEEauthorrefmark{1}
}\\
\IEEEauthorblockA{\IEEEauthorrefmark{1}University of Michigan--Dearborn}%
  \IEEEauthorblockA{\hspace{0.5ex}\IEEEauthorrefmark{2}Cleveland State University}%
}

\maketitle

\begin{abstract}
Numerous communications and networking challenges prevent deploying unmanned aerial vehicles~(UAVs) in extreme environments where the existing wireless technologies are mainly ground-focused; and, as a consequence, the air-to-air channel for UAVs is not fully covered. In this paper, a novel spatial estimation for beamforming is proposed to address UAV-based joint sensing and communications~(JSC). The proposed spatial estimation algorithm relies on using a delay-tolerant observer-based predictor, which can accurately predict the positions of the target UAVs in the presence of uncertainties due to factors such as wind gust. The solution, which uses discrete-time unknown input observers~(UIOs),  reduces the joint target detection and communication complication notably  by operating on the same device and performs reliably in the presence of channel blockage and interference. The effectiveness of the proposed approach is demonstrated using simulation results. 
\end{abstract}


\IEEEpeerreviewmaketitle

\section{Introduction}

Thanks to the technological advancements in unmanned aerial vehicles~(UAVs), a vast number of applications ranging from mapping and surveying to search and rescue in emergency temporary communications~\cite{gupta2015survey} have been enabled. Utilizing communication-equipped agile UAVs will benefit applications such as  aerial
base stations for expanding the coverage of ground cellular systems 
in 5G and beyond~\cite{orsino2017effects}. In many of these applications such as multi-UAV cinematography~\cite{alcantara2021optimal,mademlis2019autonomous}, the fast moving UAVs should maintain a reliable inter-UAV connectivity~\cite{kaadan2014multielement,goddemeier2015investigation,fabra2017methodology,amponis2021inter}. Meeting this connectivity requirement on one hand while conducting target detection and tracking on the other relies on coexistence, cooperation, and joint design of sensing and communications~(JSC).

%

The core challenges in JSC include, but are not limited to, implementing radar and communication coexistence~(RCC) system within the same band and hardware platform,  designing dual-functional radar-communication~(DFRC) systems, and combating the  communication channel  blockage due to geographical factors and extreme environments~\cite{nartasilpa2018communications,rahmati2017ssfb}.
 Despite these complications, many existing solutions assume a fully-accessible channel that is not feasible in practical scenarios. Accordingly, there is a need for high-quality joint design for UAV networks in high mobility scenarios under practical constraints such as unavailable channel information. JSC can utilize mmWave bands and multiple-input multiple-output~(MIMO) schemes for communications in order to accommodate high-throughput applications and to provide higher degrees of freedom~\cite{sarieddeen2020next}.

In contrast to conventional precoder designs for beamformers,  high-mobility UAVs encounter numerous complications for handling air-to-air channels due to the reduced channel coherence time and the resulting faster channel variation. Moreover, the channel variations lead to changes in the channel path characteristics such as path number and path direction. These channel variations require efficient training and tracking methods~\cite{ullah2020cognition}. 
If the link remains stable, full-duplex relays can be employed on UAVs to boost the capacity of mmWave networks~\cite{sanchez2020millimeter}; however, this method requires a full knowledge of the location of UAVs.  
Accordingly, measuring/predicting the location of UAVs is crucial to beamforming whether it be digital or analog, where the beams must be steered to reduce large-scale propagation attenuation. Furthermore, despite the fact that beamformer design has been thoroughly researched in the literature~\cite{friedlander2012transmit,palomar2003joint}, the design for JSC capable of meeting both the sensing and communication requirements is still 
an active research area.

Along the advances in JSC and mmWave UAV communications, the field of control systems is witnessing a recent advent of algorithms for large-scale collaborative sensor and actuator networks. Indeed, there is a re-emergence of interest in designing state observers/estimators with the ultimate objective of  achieving state-omniscience~\cite{han2018simple,Moham2017} at a central node and/or a collection of distributed nodes across these sensor/actuator networks, where the goal of such algorithms is to enable a designated network of nodes to reconstruct the full dynamical state of the network. Therefore, the issue of information delay and estimation of 
delayed data, whether the delay be network-induced or long distance-induced, is of crucial importance in the modern communications networks. Discrete-time unknown input observers (UIOs)  are among the prominent tools for estimating the states and unknown inputs of a given dynamical system in the presence of delays~\cite{floquet2006state,besanccon2015control,chakrabarty2017delayed,hur2014unknown}.


In this article, we will employ the UIO architecture in~\cite{chakrabarty2017delayed} for predicting the location of a collection of UAVs from the previous location of them and estimating the unknown inputs acting on the UAVs in the presence of time-delays.  The unknown input observer in~\cite{chakrabarty2017delayed} can provide a guaranteed level of attenuation for bounded unknown inputs with any specified maximum estimation delay. From this perspective, the delayed UIO in~\cite{chakrabarty2017delayed} has desirable properties for our beamforming application.  Indeed, the proposed UIO essentially plays the role of a \emph{virtual radar} for an efficient and low-cost implementation of JSC providing an accurate prediction of the UAV positions for the beamformer within the communication waveform. Despite the existence of promising mmWave FMCW radar technologies for localizing UAVs (see, e.g.,~\cite{rizzoli2009new,rai2021localization}), our proposed UIO-based solution for localizing the UAVs becomes more desirable in applications where there exist strict power consumption constraints such as those on the central UAV.   Furthermore, when the channel is not fully available, our proposed method provides a robust solution for beamforming in the presence of steering angle errors as well as  combining sensing and communications in JSC systems. 

To the best of our knowledge, the only lines of work that provide a solution for JSC using UAV location prediction is due to Liu \emph{et al.}~\cite{liu2020location} and Peng \emph{et al.}~\cite{peng2018unified}. In the former, a deep learning-based approach is proposed in order to feed a beamformer while in the latter a Kalman filter is used for UAV location prediction.  These approaches differ from ours in several aspects. First, the works in~\cite{liu2020location,peng2018unified} assume known statistical distributions for the UAV velocity and the external unknown inputs causing perturbations in the nominal velocity profiles of the UAVs. In our work, we do not make any assumptions on these statistics, thanks to our UIO-based approach. Second, the LRNet architecture in~\cite{liu2020location} is used for predicting the location of a single UAV and the scalability of this approach to a network of UAVs is yet to be tested. In contrast, our approach can  be applied to any number of UAVs. 



The rest of the paper is organized as follows. In Sect.~\ref{sec:sysmodel}, we introduce our system model and discuss our solution in details. 
In Sect.~\ref{sec:observ}, we propose our mathematical framework and analysis for the central observer design for UAV networks.
In Sect.~\ref{sec:perf} we evaluate the performance of our solution using simulations and show the amount of gain we obtain compared to other existing techniques. Finally, in Sect.~\ref{sec:conc}, we conclude the paper with final remarks and future research directions. \vspace{0.1ex}

\noindent\textbf{Notation.} 
Given a matrix $\mathbf{M}$, we denote its transpose by $\mathbf{M}^{\top}$, its Hermitian by $\mathbf{M}^{H}$, its inverse by $\mathbf{M}^{-1}$, and its conjugate by $\mathbf{M}^*$, respectively. We let $\text{diag}(a_1,\cdots,a_N)$ denote the $N\times N$ diagonal matrix with $a_1,\cdots,a_N$ on its diagonal. For any vector $\mathbf{a}\in\mathbb{R}^N$, we denote $\|\mathbf{a}\| = \sqrt{\mathbf{a}^\top \mathbf{a}}$.  For any matrix $\mathbf{A}$, we let $\|\mathbf{A}\|$ denote the maximum singular value of $\mathbf{A}$. By $\mathbf{A} \succeq \mathbf{B}$ for two symmetric matrices $\mathbf{A}$ and $\mathbf{B}$, we mean $\mathbf{A}-\mathbf{B}$ is a positive semi-definite matrix. For a sequence of vectors $\{\mathbf{a}_k\}_{k=0}^\infty$, we denote $\|\mathbf{a}\|_\infty := \sup_{k\geq 0} \|\mathbf{a}_k\|$. We say a sequence $\{\mathbf{a}_k\}\in \ell _\infty$ if $\|\mathbf{a}\|_\infty$ is bounded. For  $\{\mathbf{a}_k\}_{k=0}^\infty$, we let $\limsup_{k\to \infty} \|\mathbf{a}_k\|$ denote the limit superior of the sequence. 

\begin{figure*}[!t] 
\begin{center}
\begin{minipage}{0.99\textwidth} 
\begin{center}
\includegraphics[width=1\textwidth]{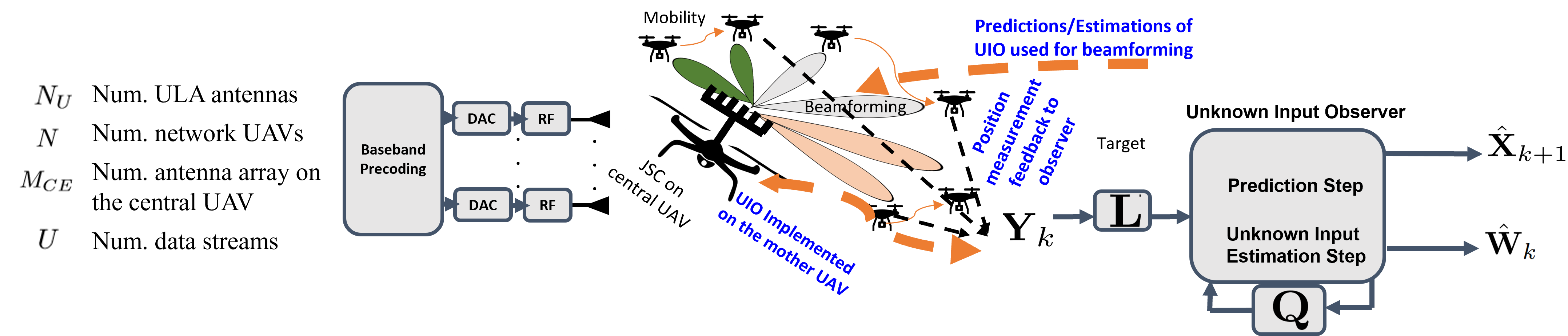}
\end{center}
\end{minipage}
\end{center}
\vspace{-0.05cm}
\caption {The schematics of the overall UIO-based location-aware beamforming algorithm.}\label{fig:sysmodel}
\vspace{-0.2cm}
\end{figure*}
\section{System Model and the Proposed Solution} \label{sec:sysmodel}

\noindent\textbf{Communication Model.} We consider a network of UAVs (see~Fig.\ref{fig:sysmodel}), which consists of a central UAV with $M_{\text{CE}}$ antennas in a massive MIMO structure communicating with N moving UAVs each with $N_{\text{U}}$ Uniform Linear Array~(ULA) antennas on the other side of the links in which $N_{\text{U}}\leq M_{\text{CE}}$ and the antenna elements are equally spaced at half wavelength (see Figure~\ref{fig:sysmodel}). 
In practical scenarios, when the communication channel is available, the line of sight~(LoS) between the transmitter and receiver is the dominant propagation component; otherwise, we assume that the blockage in the channel occurs while  affecting the beamformer. We assume the transmitter can send \textit{at least} one data stream per UAV, i.e., $N$ streams, simultaneously.  
Let the vector $\mathbf{s}_{i,k}\in \mathbb{C}^{M_{CE}}$ be the transmitted data from the central UAV to the $i$\textsuperscript{th} UAV, $1 \leq i \leq N$,
at time instant $k$, $k\geq 1$. 
%
As demonstrated in Section~\ref{sec:observ}, we only assume that the  kinematics/geometric information of the targets 
are received by the central UAV in a cooperative manner while no reliance on the radar-based echoed positions/velocities is assumed. This is in contrast to the conventional methods where the known signals are sent; and, the reflected echo waves are observed subsequently. 

We assume Saleh-Valenzuela channel model~\cite{saleh1987statistical} for mmWave communications and the received signal model at each UAV, denoted by vector $\mathbf{r}_{i,k}$, $1 \leq i \leq N$, at time instant $k$ where $\mathbf{r}_{i,k}\in \mathbb{C}^{N_{U}}$ and
is given by 
%
%
\begin{equation}
    \mathbf{r}_{i,k} = \frac{1}{\sqrt{M_{CE}N_U}} h_{i,k} \mathbf{b}(\theta_{i,k}) \mathbf{a}^H(\theta_{i,k}) \mathbf{s}_{i,k}+\mathbf{\nu}_{i,k}
    \label{eq:channelmodel}
\end{equation}
\noindent where $\mathbf{\nu}_{i,k} \in \mathbb{C}^{N_{U}}$ denotes the Gaussian noise vector with zero mean and variance $\sigma^2$ for each sample and $h_{i,k}$ stands for the channel coefficient of the $i^\text{th}$ UAV at time instant $k$. Here, for a ULA, $\mathbf{a}(\theta_{i,k})$ and $\mathbf{b}(\theta_{i,k})$  
stand for the normalized steering vectors for Angle of Departure~(AoD) at the central UAV and Angle of Arrival~(AoA) at the receiving UAV, respectively as shown below.
\begin{equation}
\begin{aligned}
&\text{Steering vector on central UAV at time instant $k$:\,}\\
  & \mathbf{a}(\theta_{i,k}) =  [1,e^{j{\frac{2\pi}{\lambda}d} \sin(\theta_{i,k})},...,e^{j{\frac{2\pi}{\lambda}d} (M_{\text{CE}}-1) \sin(\theta_{i,k})}]^\top \in \mathbb{C}^{M_{\text{CE}}}\\
  &\text{Steering vector on moving UAVs at time instant $k$:\,}\\
  &  \mathbf{b}(\theta_{i,k}) =  [1,e^{j{\frac{2\pi}{\lambda}d} \sin(\theta_{i,k})},...,e^{j{\frac{2\pi}{\lambda}d} (N_{\text{U}}-1) \sin(\theta_{i,k})}]^\top \in \mathbb{C}^{N_{\text{U}}}\\
    \end{aligned}
    \label{eq:steeringtransm}
\end{equation}
\noindent where $d$ and $\lambda$ are the antenna spacing and wavelength, respectively. The angles $\theta_{i,k}$, at each time instant $k$ are determined by the location of the UAVs, whose dynamics are governed by a collection of kinematic equations of motion (see ``Equations of Motion of the UAV Network''). 

\begin{remark}
It should be noted that with traditional radar systems, undesirable scatters would interfere with the received signal. We do not have such an interference component in the received signal in the proposed technique in this research since the sensing is performed in a novel fashion, as described in Section~\ref{sec:observ}.
\end{remark}

\noindent\textbf{Equations of Motion of the UAV Network.} In this article, we use the dynamical model adopted by the UAV-based communications literature (see, e.g., ~\cite{peng2018unified}). Similar kinematic models have also been utilized in UAV flocking and MPC-based control 
for autonomous UAV landing (see, e.g.,~\cite{chapman2011uav,feng2018autonomous}). We denote the position of the $i$\textsuperscript{th} moving UAV at time instant $k$,  by $\mathbf{u}_{i,k}:=[x_{i,k},\, y_{i,k}]^\top$, where $1 \leq i \leq N$, and $k \geq 0$. Then, the equations of motion read as
\begin{equation}
    \mathbf{u}_{i,k+1} = \mathbf{u}_{i,k} + \Delta T_i \mathbf{v}_{i,k} + \mathbf{d}_{i,k}, 
    \label{eq:singleDrone}
\end{equation}
where $\mathbf{v}_{i, k}=[v_{i, k}^x,\, v_{i, k}^y]^\top$ and $\mathbf{d}_{i, k}=[d_{i, k}^x,\, d_{i, k}^y]^\top$ denote the UAV average nominal velocity and average velocity perturbation due to unknown inputs such as wind gust at time instant $k$, respectively. Each time duration $\Delta T_i$, $1\leq i \leq N$, represents  the length of the time interval for updating the position measurements associated with the $i$\textsuperscript{th} UAV. In this article, we allow these update time intervals to be different. \\
\indent We concatenate the UAV positions, velocities, and unknown inputs at each time instant $k$ into the single vectors $\mathbf{X}_k := [\mathbf{u}_{1,k}^\top, \cdots, \mathbf{u}_{N,k}^\top]^\top \in \mathbb{R}^{2N}$, $\mathbf{V}_k := [\mathbf{v}_{1,k}^\top, \cdots, \mathbf{v}_{N,k}^\top]^\top \in \mathbb{R}^{2N}$, $\mathbf{\Lambda}_k := [\mathbf{d}_{1,k}^\top, \cdots, \mathbf{d}_{N,k}^\top]^\top \in \mathbb{R}^{2N}$, respectively. Then, the aggregate equation of motion of the UAV network is 
\begin{equation} 
\mathbf{X}_{k+1} = \mathbf{I}_{2N} \mathbf{X}_{k} + \mathbf{B}_T \mathbf{V}_{k} + \pmb{\Lambda}_{k},
\end{equation}
where $\mathbf{B}_T=\text{diag}\big(\Delta T_1, \Delta T_1, \cdots, \Delta T_N, \Delta T_N\big)\in \mathbb{R}^{2N\times 2N}$ is the diagonal matrix of sampling times. In this state space equation, the vector $\mathbf{X}_k \in \mathbb{R}^{2N}$ denotes the 
state vector of the UAV network; namely, the position of the UAVs.

\indent We define the vector of unknown inputs as  
\begin{equation}
\mathbf{W}_{k} := \mathbf{V}_{k}+ \mathbf{B}_T^{-1} \pmb{\Lambda}_{k}. 
\end{equation}
This vector, which lumps the effect of the nominal velocities of the UAVs along with perturbations due to uncertainties such as wind gust, drives the dynamics of the UAVs. In the absence of uncertainties, we have $\mathbf{W}_{k} = \mathbf{V}_{k}$.  Hence, the discrete-time dynamics in~\eqref{eq:singleDrone} in the aggregate are governed by 
\begin{equation}
\begin{aligned}
    \mathbf{X}_{k+1} = \mathbf{I}_{2N} \mathbf{X}_{k} + \mathbf{B}_T \mathbf{W}_{k}, \\
    \mathbf{Y}_{k} = \mathbf{I}_{2N} \mathbf{X}_{k} + \mathbf{D} \mathbf{W}_{k}. 
    \end{aligned}
    \label{eq:multiDrone}
\end{equation}
In the state space model given by~\eqref{eq:multiDrone}, $\mathbf{Y}_{k}$ represents the vector of UAV position measurements that are affected by 
uncertainties. If the uncertainties do not affect the 
UAV vector measurements, then the matrix $\mathbf{D}\in \mathbb{R}^{N \times 2N}$ can be set equal to $\mathbf{0}$. According to~\eqref{eq:multiDrone}, the UAV position signals that are received by the UIO do not need to be accurate and can be affected by the vector of uncertainties $\mathbf{W}_{k}$. 

\vspace{-0.15cm}
\begin{remark}
\label{rem:adsb}
We are assuming that either the smaller UAVs are transmitting their positions to the central UAV via an automatic dependent surveillance-broadcast (ADS-B) system~\cite{de2013using,moody2009implementation} or their positions can be inferred using a vision system on-board the central UAV (see, e.g., the method in~\cite{kendoul2009optic}). 
\end{remark}
\vspace{-0.15cm}
\begin{remark}
Unlike~\cite{liu2020location,peng2018unified}, we make no
assumptions on the statistics of $\pmb{\Lambda}_{k}$ and $\mathbf{V}_k$.  In Section~\ref{sec:observ}, we present an unknown input observer (UIO) algorithm for predicting the future location of the UAVs by only using the approximate UAV position measurements given by $\mathbf{Y}_k$ without relying on the knowledge of the vector of  unknown inputs $\mathbf{W}_{k}$. 
\end{remark}
\vspace{-0.15cm}

\noindent\textbf{Beamformer Design.} Given an estimation of the angular position vector of the UAVs in the network  $\hat{\pmb{\theta}}_k = [\hat{\theta}_{1,k},\cdots, \hat{\theta}_{N,k}]^\top$ at time instant $k$ (see Section~\ref{sec:observ} for this estimated value), the baseband transmit beamformer $\mathbf{F}$ is given by 
\begin{equation}
    \mathbf{F} =  \mathbf{A}^*(\hat{\pmb{\theta}}_{k})\Big(\mathbf{A}^\top(\hat{\pmb{\theta}}_{k})\mathbf{A}^*(\hat{\pmb{\theta}}_{k})\Big)^{-1} \in \mathbb{C}^{M_{\text{CE}} \times N},
    \label{eq:beamformer}
\end{equation}
where $\mathbf{A}(\hat{\pmb{\theta}}_{k})=[\mathbf{a}(\hat{\theta}_{1,k}),\mathbf{a}(\hat{\theta}_{2,k}),...,\mathbf{a}(\hat{\theta}_{N,k})] \in \mathbb{R}^{M_\text{CE}\times N}$ is the steering matrix at the central UAV for communicating with $N$ moving UAVs (see Figure~\ref{fig:sysmodel}). In this case, $\hat{\mathbf{s}}\in \mathbb{C}^{N}$ represents the vector of transmitting data to UAVs, which is beamformed by ${\mathbf{s}}=\mathbf{F}\hat{\mathbf{s}}$ and the result is used in~\eqref{eq:channelmodel}. In Section~\ref{sec:observ}, we present an observer-based algorithm that can estimate the angular positions in~\eqref{eq:beamformer}. 


\vspace{-0.1cm}
\begin{remark}
In this paper, we are assuming that  the difference in the AoA of $N_U$ antennas in each UAV is negligible due to the small size of the UAVs. Furthermore, because of the dominant LOS component, there are very few reflected paths resulting in a sparse channel. Similar assumptions have been made in~\cite{wymeersch20185g} for mmWave positioning of vehicular networks.
\end{remark}
\vspace{-0.1cm}
\section{The Centralized UIO for UAV Networks} \label{sec:observ}

We adopt the predictor architecture due to Chakrabarty \emph{et al.}~\cite{chakrabarty2017delayed}, where the predictor can estimate the unknown inputs acting on a given discrete time dynamics and simultaneously predict the next location of the UAVs. In the context of our work, the unknown inputs are the average velocities of the UAVs and the vector of uncertainties.  The dynamics of the observer/predictor in~\cite{chakrabarty2017delayed} when employed for the UAV network dynamics in~\eqref{eq:multiDrone} are governed by 
\begin{equation}
\text{Prediction step:\,}    \hat{\mathbf{X}}_{k+1} = \mathbf{Q} \hat{\mathbf{X}}_{k}  + \mathbf{L} \mathbf{Y}_{k}, 
    \label{eq:observer}
\end{equation}
\noindent where $\mathbf{Q},\, \mathbf{L}\in\mathbb{R}^{2N\times2N}$ are called the observer state and observer gain matrices, respectively. The observer state and gain matrices are designed by solving proper linear matrix inequalities (LMIs). Furthermore, we define the following performance output vector
\begin{equation}
    \hat{\mathbf{Z}}_{k} = \mathbf{H} \mathbf{E}_k,
    \label{eq:perfVector}
\end{equation}
where $\mathbf{E}_k := \hat{\mathbf{X}}_{k} - \mathbf{X}_{k}$ is the state tracking error, and 
the design parameter $\mathbf{H}$ is the performance output matrix to select subset/linear combinations of error states that are crucial to the specific application. If all position estimation errors are of the same importance, then one can choose $\mathbf{H}=\mathbf{I}_{2N}$. 
\begin{remark}
In this paper, we design UIOs for predicting the trajectories of UAVs in two-dimensional settings. The developed UIO-based prediction framework can be easily extended to three-dimensional scenarios.
\end{remark}

According to~\cite{chakrabarty2017delayed}, to find the state matrix $\mathbf{Q}$ and the gain matrix $\mathbf{L}$, the following LMI system needs to be solved  \\
\begin{equation}
\begin{aligned}
    \begin{bmatrix} (\alpha-1) \mathbf{P} & \star & \star \\
    \mathbf{0} & -\alpha \mathbf{I}_{2N} & \star \\
    \mathbf{P} - \mathbf{Z} & \mathbf{Z} \mathbf{D} - \mathbf{P} \mathbf{B}_T & \mathbf{P}
    \end{bmatrix}  \preceq \mathbf{0},\, 
    \begin{bmatrix} \mathbf{P} & \mathbf{H}^\top \\
    \mathbf{H} & \mu \mathbf{I}_{2N}    \end{bmatrix} \succeq \mathbf{0},
    \end{aligned}
    \label{eq:LMI}
\end{equation}
\noindent where $\alpha \in (0,\, 1)$ is fixed, and the LMI is solved for $\mathbf{P}=\mathbf{P}^\top \succ \mathbf{0}$, $\mathbf{Z}$, and some desired performance level $\gamma = \sqrt{\mu} > 0$ (to be defined in the next proposition). The $\star$ in~\eqref{eq:LMI} is used to avoid rewriting the symmetric terms. Having found the matrices $\mathbf{P}$ and $\mathbf{Z}$ from the LMIs in~\eqref{eq:LMI}, we use   
\begin{equation}
\mathbf{L}=\mathbf{P}^{-1}\mathbf{Z},\; 
\mathbf{Q}=\mathbf{I}_{2N} - \mathbf{L}, 
\label{eq:LQfound}
\end{equation}
to obtain the observer state and gain matrices, respectively. 

\begin{remark}
The size of the LMI in~\eqref{eq:LMI} is proportional to the number of the UAVs in the network. An efficient numerical algorithm for for finding solutions to the LMI in~\eqref{eq:LMI} is due to Gahinet and Nemirovski~\cite{gahinet1997projective}. The algorithm in~\cite{gahinet1997projective}, which also lies at the heart of MATLAB LMI solvers, enjoys a polynomial-time complexity. 
\end{remark}

According to the observer architecture due to Chakrabarty \emph{et al.}~\cite{chakrabarty2017delayed}, it is possible to utilize the output of the predictor in~\eqref{eq:observer} and the measurements given by  $\mathbf{Y}_k$ in~\eqref{eq:multiDrone}, to find an estimate of the unknown input vector  $\hat{\mathbf{W}}_k$ according to 
\begin{equation}
  \text{\begin{minipage}{0.15\textwidth}
  Unknown input \\
  estimation step:
  \end{minipage}}  \hat{\mathbf{W}}_k = \mathbf{G}^{\dagger} \begin{bmatrix} \hat{\mathbf{X}}_{k+1} - \hat{\mathbf{X}}_{k} \\
    \mathbf{Y}_k - \hat{\mathbf{X}}_{k}\end{bmatrix}
    \label{eq:inputEstimate}
\end{equation}
where $\mathbf{G}^{\dagger}= (\mathbf{G}^\top \mathbf{G})^{-1} \mathbf{G}^\top$ is the the Moore–Penrose pseudo-inverse of the matrix $\mathbf{G}:= [\mathbf{B}_T^\top,\, \mathbf{0}_{2N}^\top]^\top$. Since $\mathbf{G}$ is of full column rank, i.e., $\text{rank}(\mathbf{G}) = 2N$, the Moore-Penrose pseudo-inverse of $\mathbf{G}$ is well-defined. 

We have the following propositions regarding the state tracking error  and unknown input reconstruction performance of the observer architecture in~\eqref{eq:observer} and~\eqref{eq:inputEstimate}. 
\begin{proposition}\label{prop:state}
Consider the observer/predictor architecture in~\eqref{eq:observer}. Given the performance matrix $\mathbf{H}$ and a fixed constant $\alpha \in (0,1)$, there exists a bounded performance level  $\gamma = \sqrt{\mu} > 0$ proportional to the norm $\|\mathbf{H}\|$, such that the system of LMIs in~\eqref{eq:LMI} has a solution for the unknown matrices $\mathbf{P}=\mathbf{P}^\top \succ \mathbf{0}$, $\mathbf{Z}$. Furthermore, the observer in~\eqref{eq:observer} with state and gain matrices given by~\eqref{eq:LQfound} is $\ell_\infty$-stable with performance level $\gamma$ with respect to the unknown input $\mathbf{W}_k$. Equivalently, the following properties are satisfied. 
\begin{itemize}
    \item[i] If $\mathbf{W}_k=\mathbf{0}$, then the origin of the state tracking error dynamics are globally uniformly exponentially stable;
    \item[ii] For zero initial error $\mathbf{E}_0$ and every bounded unknown input sequence $\{\mathbf{W}_k\}$, we have $\|\mathbf{Z}_k\|\leq \gamma \|\mathbf{W}\|_\infty$, where $\|\mathbf{W}\|_\infty := \sup_{k\geq 0} \|\mathbf{W}_k\|$;  
    \item[iii] For every initial error $\mathbf{E}_0$ and every bounded unknown input $\{\mathbf{W}_k\}$, we have $\limsup_{k\to \infty} \|\mathbf{Z}_k\|\leq \gamma \gamma_{_W}$, where  $\gamma_{_W}:=\|\mathbf{W}\|_\infty$ is the tightest upper bound on the sequence of unknown inputs.
\end{itemize}
\end{proposition}
\vspace{-0.35cm}
\begin{proof}
The pair of system matrices $(\mathbf{I}_{2N},\, \mathbf{I}_{2N})$, resulting from the dynamics of the UAV network in~\eqref{eq:multiDrone}, is detectable. Therefore, from the proof of Theorem~3 in~\cite{chakrabarty2017delayed}, it immediately follows that  there exists a bounded performance level  $\gamma = \sqrt{\mu} > 0$ proportional to the norm $\|\mathbf{H}\|$, such that the system of LMIs in~\eqref{eq:LMI} has a solution. Since the solution to the LMIs in~\eqref{eq:LMI} exists, it follows from Theorem~2 in~\cite{chakrabarty2017delayed} that the observer in~\eqref{eq:observer} with state/gain matrices given by~\eqref{eq:LQfound} satisfies the $\ell_\infty$-stability property claimed by the proposition.
\end{proof}

The following proposition provides an upper bound on unknown input estimation error. 

\begin{proposition}\label{prop:input}
Suppose the assumptions in Proposition~\ref{prop:state} hold and 
$\mathbf{H}={I}_{2N}$. Then, the unknown input reconstruction error is bounded by  $\limsup_{k\to \infty} \|\hat{\mathbf{W}}_k - \mathbf{W}_k\|\leq 3 \gamma \gamma_{_W}$, where $\gamma_{_W}:=\|\mathbf{W}\|_\infty$ is the tightest upper bound on the sequence of unknown inputs. 
\end{proposition}
\begin{proof}
The proof follows from Theorem~5 in~\cite{chakrabarty2017delayed}. 
\end{proof}

To compute the steering vectors that are fed to the beamformer in~\eqref{eq:beamformer}, we utilize the position estimates provided by the UIO in~\eqref{eq:observer}. In particular, the angular position vector $\theta_{i,k}$ of the $i$\textsuperscript{th} UAV at the $k$\textsuperscript{th}  time-instant can be easily computed from $\hat{\mathbf{X}}_k=[\hat{\mathbf{u}}_{1,k}^\top, \cdots, \hat{\mathbf{u}}_{N,k}^\top]^\top$ and  the equation 
\begin{equation}
   \hat{\theta}_{i,k} = \text{arccos} \frac{\hat{x}_{i,k} - x_p}{\|\hat{\mathbf{u}}_{i,k} - \mathbf{u}_p\|},
   \label{eq:angPos}
\end{equation}
\noindent where $\mathbf{u}_p$ is the position of the central UAV. The angles in~\eqref{eq:angPos} can be used for constructing the steering matrix in~\eqref{eq:beamformer}. \\

\noindent{\textbf{Summary of the Proposed Location-Aware Beamformer.}} The proposed location-aware beamformer can be summarized as follows (also, see Figure~\ref{fig:sysmodel}). The UIO dynamics in~\eqref{eq:observer} and~\eqref{eq:inputEstimate}, where the UIO observer and state gain matrices are found by solving~\eqref{eq:LMI}, can be used for estimating the correct position of the UAVs in the network along with their unknown velocity profiles. The estimated positions can then be used in~\eqref{eq:angPos} to compute the angular positions of the UAVs. The estimated angular positions of the beamformers can in turn be utilized in~\eqref{eq:beamformer} and as a result, the model in~\eqref{eq:channelmodel} will be updated with the beamformed data towards the desired UAV.

\section{Simulation Studies} \label{sec:perf}
\begin{figure*}[!t] 
\begin{minipage}{0.49\textwidth}
\includegraphics[width=0.99\textwidth]{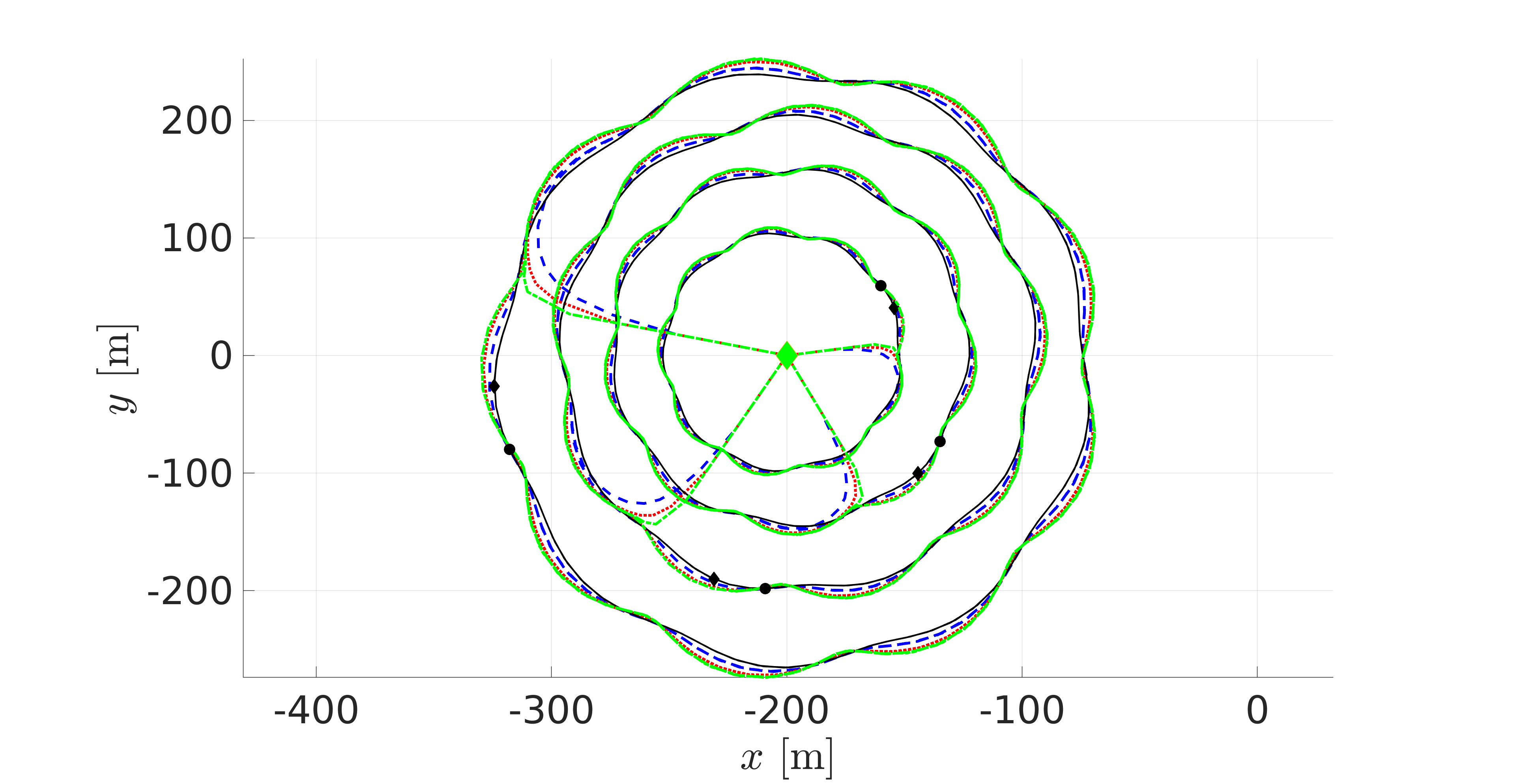}
\end{minipage}
\begin{minipage}{0.49\textwidth}
\hspace{-0.4cm}
\includegraphics[width=1.15\textwidth]{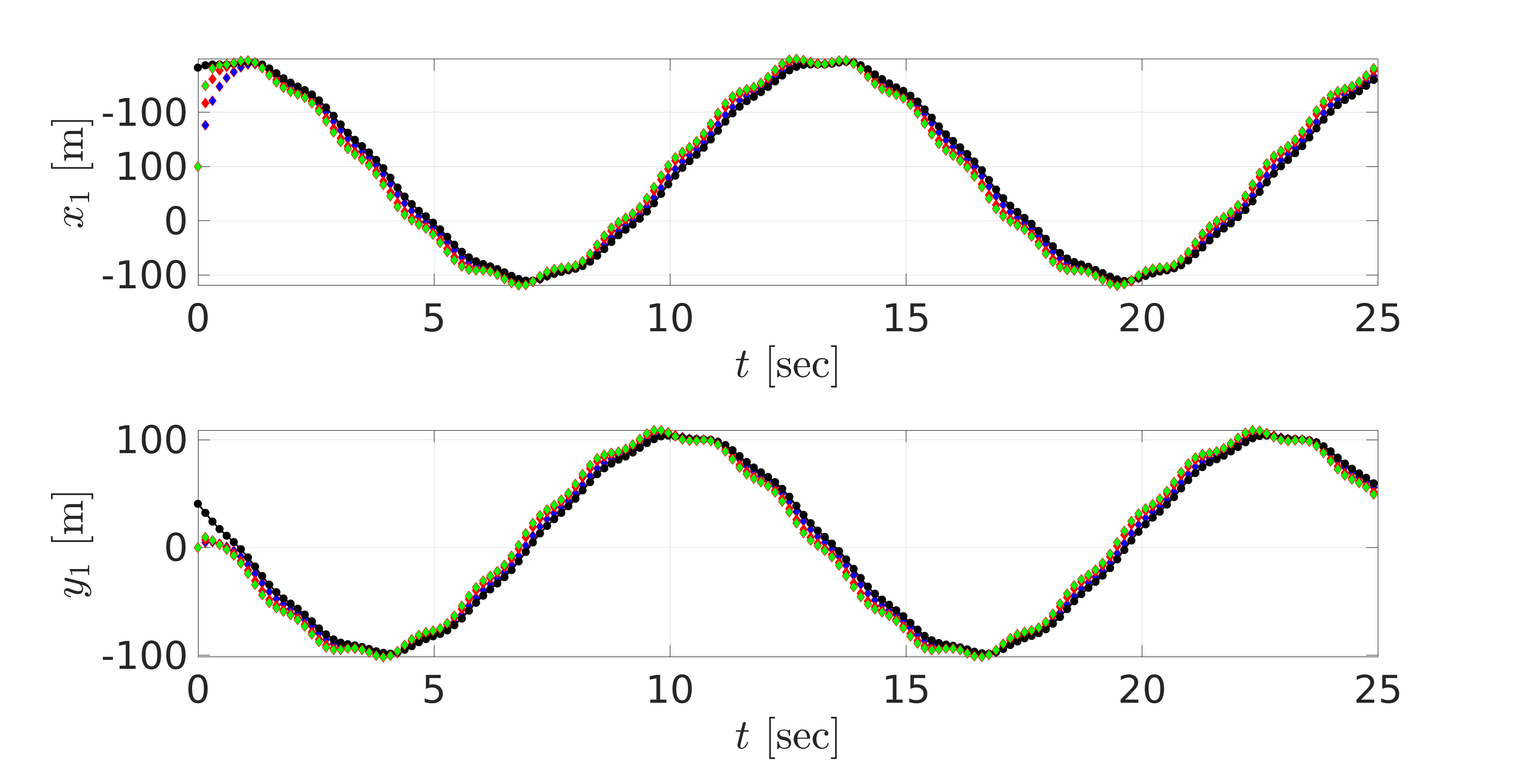}
\end{minipage}
\vspace{-0.05cm}
\caption {Location prediction simulation results with design parameter $\alpha=0.5$ in~\eqref{eq:LMI}: (top left) UAV network actual/predicted paths; (top right) time profile of the coordinates of UAV\#1. The simulation curve colors are associated with the performance levels $\gamma_1=0.21$ (blue), $\gamma_2=0.47$ (red), and $\gamma_3=0.96$ (green), respectively. The black curves in the top left and the top right plots belong to the actual path of the UAVs and time profile of the position coordinates of UAV\#1.}\label{fig:simResultsLocation}
\vspace{-0.1cm}
\end{figure*}
In this section we present the simulation results to verify the effectiveness of our proposed methodology. In the simulations, we assume that base drone is hovering at the origin $\mathbf{x}_p = [0,\, 0]^\top$. We consider $N\text{= }4$ UAVs that are moving on closed curves centered at the base UAV (see Fig.~\ref{fig:simResultsLocation}). The duration for time instants associated with each of the four UAVs are considered to be $\Delta T\text{= }0.15^\text{sec}$. 

It is assumed that the UAVs, which are 
flying around the base UAV, traverse perturbed circular paths that are affected by sinusoidal velocity perturbations. It is remarked that sinusoidal motion perturbations have also been  utilized in other UAV simulation studies as well (see, e.g., \cite{bencatel2011formation,tonetti2011distributed}). In \eqref{eq:singleDrone}, the nominal velocity of the $i$\textsuperscript{th} UAV is given by  
\begin{equation}
\mathbf{v}_{i,k} = R_i \omega \begin{bmatrix} \cos(\omega \Delta T k + \phi_i) \\ -\sin(\omega \Delta T k + \phi_i) \end{bmatrix}, 
\end{equation}
and the velocity perturbation of the $i$\textsuperscript{th} UAV is given by  
\begin{equation}
\mathbf{d}_{i,k} = \frac{R_i \omega}{5}\begin{bmatrix} \cos(10\omega \Delta T k + \phi_i) \\ -\sin(10\omega \Delta T k + \phi_i) \end{bmatrix}, 
\end{equation}
where $R_i$ is the radius of the nominal circle that the $i$\textsuperscript{th} UAV is flying on and $\omega = 0.5 \text{ rad/sec}$. 

The UIO neither knows the nominal velocity of the UAVs nor does it know the sinusoidal velocity perturbations exerted on the UAVs. As discussed in Section~\ref{sec:sysmodel}, we are assuming that either the smaller UAVs are transmitting their positions to the central UAV via an ADS-B system or their positions can be inferred using a vision system on-board the central UAV (see Remark~\ref{rem:adsb}). However, the received position signals are not accurate and are perturbed through the unknown input $\mathbf{W}_{k}$. In other words, the only information available to the UIO is through the perturbed position signals, namely, through $\mathbf{Y}_{k} = \mathbf{I}_{8} \mathbf{X}_{k} + \mathbf{D} \mathbf{W}_{k}$. In these simulations, $\mathbf{D}=0.5\mathbf{I}_8$.  

\begin{figure*}[!t] 
\begin{minipage}{0.49\textwidth}
\hspace{-0.5cm}
\includegraphics[width=1.15\textwidth]{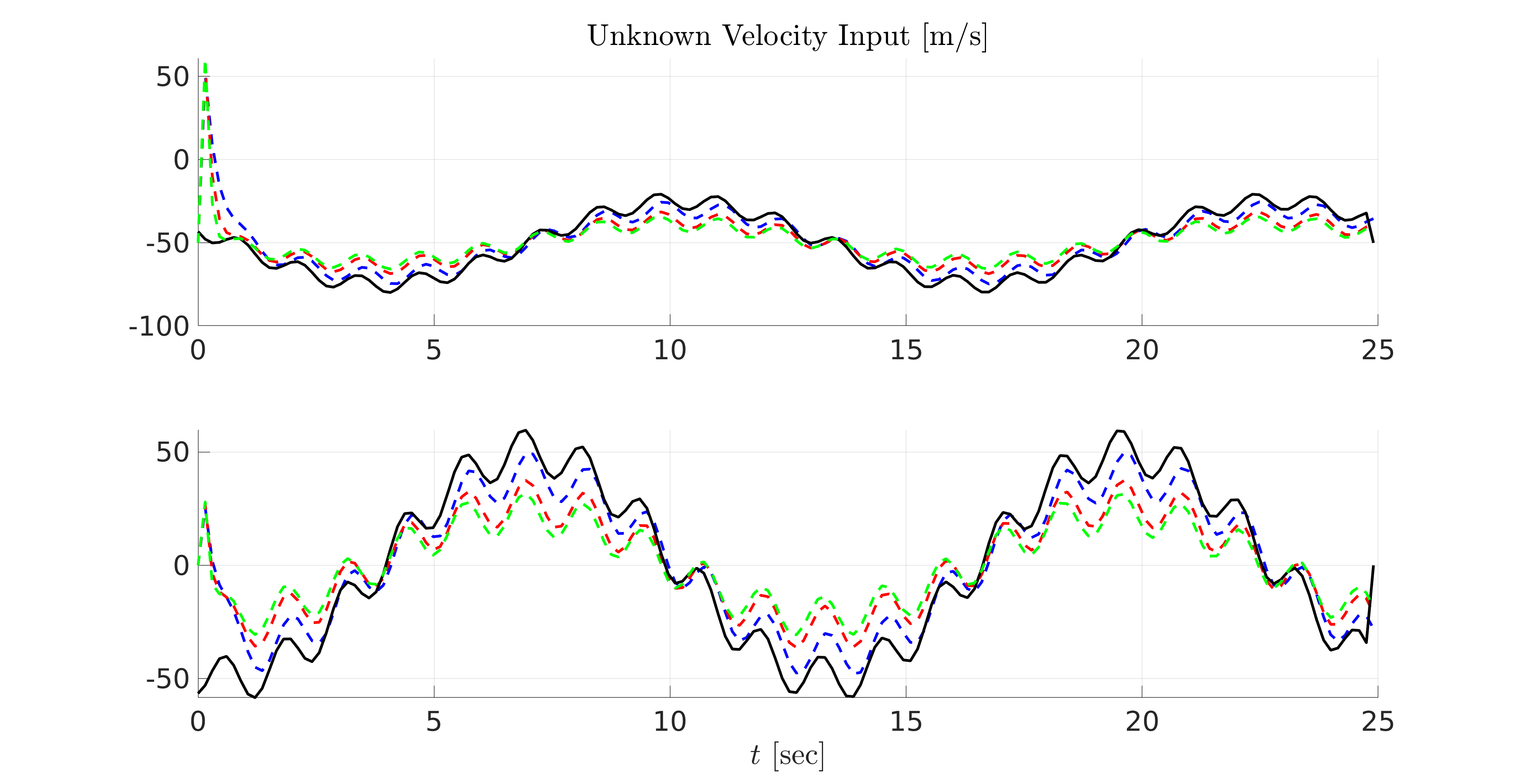}
\end{minipage}
\begin{minipage}{0.49\textwidth}
\hspace{-0.5cm}
\includegraphics[width=1.15\textwidth]{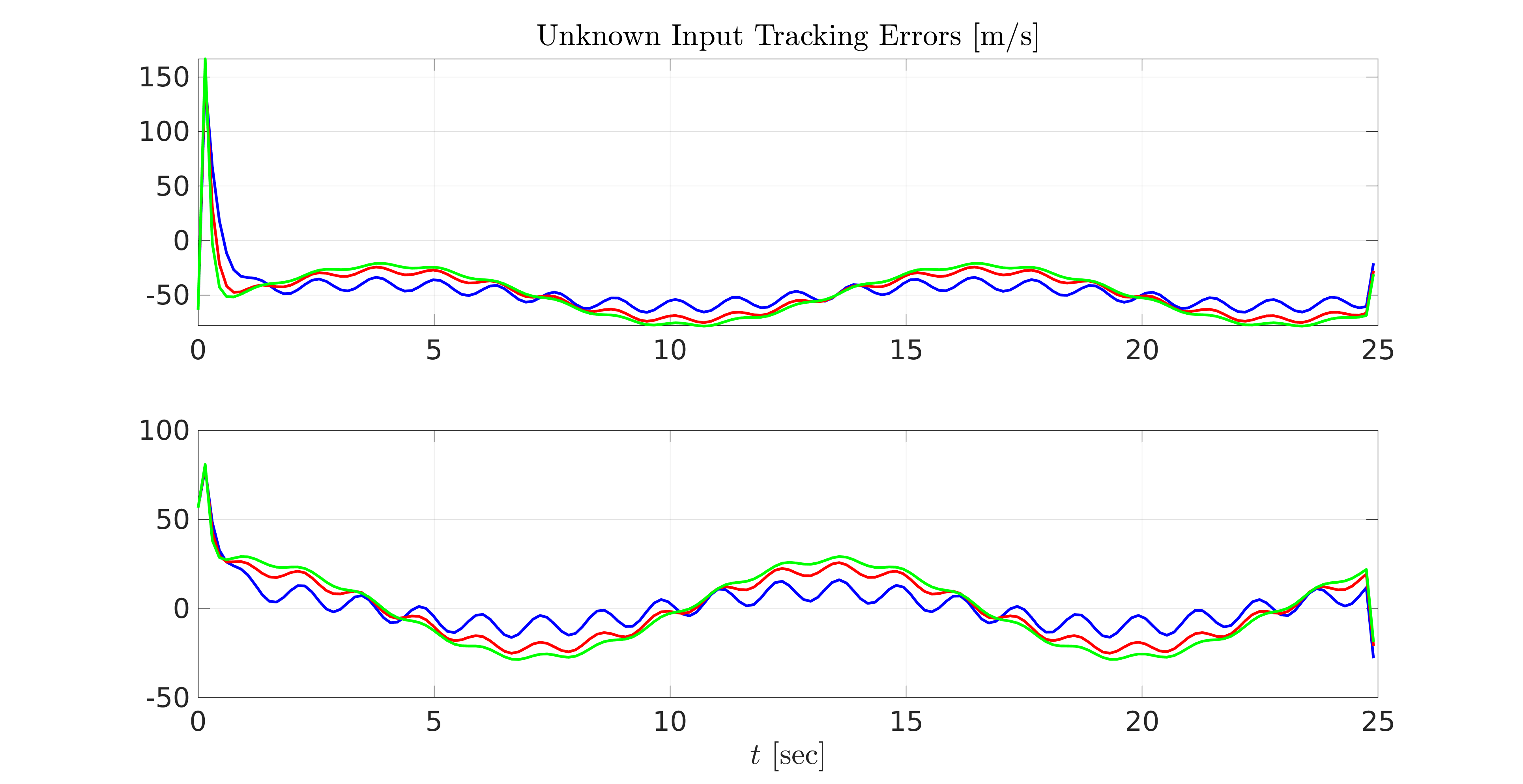}
\end{minipage}
\vspace{-0.1cm}
\caption {Unknown input tracking error results with design parameter $\alpha=0.5$ in~\eqref{eq:LMI}: (left) estimated/actual unknown inputs acting on UAV\#1; and (right) unknown input tracking errors for UAV\#1. The curve colors are associated with the performance levels $\gamma_1=0.21$ (blue), $\gamma_2=0.47$ (red), and $\gamma_3=0.96$ (green), respectively. The black curves in the left plot belong to the actual unknown input acting on UAV\#1.}\label{fig:simResultsLocation2}
\vspace{-0.3cm}
\end{figure*}

We consider a $64$-antenna MIMO on the central UAV and a $4$-antenna ULA on the other UAVs, respectively, unless otherwise specified. The maximum range is assumed to be less than $500~\rm{m}$ for an efficient mmWave communication such that the UAVs are separated from the centeral UAV by $100~\rm{m}$, $150~\rm{m}$, $200~\rm{m}$, and $250~\rm{m}$, respectively. The specified angles are extracted from the UIO presented in Section~\ref{sec:observ}.
A 16-QAM baseband modulation within Orthogonal Frequency Division Multiplexing~(OFDM) data symbols with $234$ subcarriers and 
cyclic prefix length of $64$ is considered to 
generate the waveform which operates at the central frequency of $30{\text{GHz}}$ with the baseband bandwidth of $50{\text{MHz}}$. A $1/3$ Convolutional code is utilized to protect the stream against channel errors.

\begin{figure*}[!t] 
\begin{minipage}{0.49\textwidth}
\hspace{-0.5cm}
\includegraphics[width=1.15\textwidth]{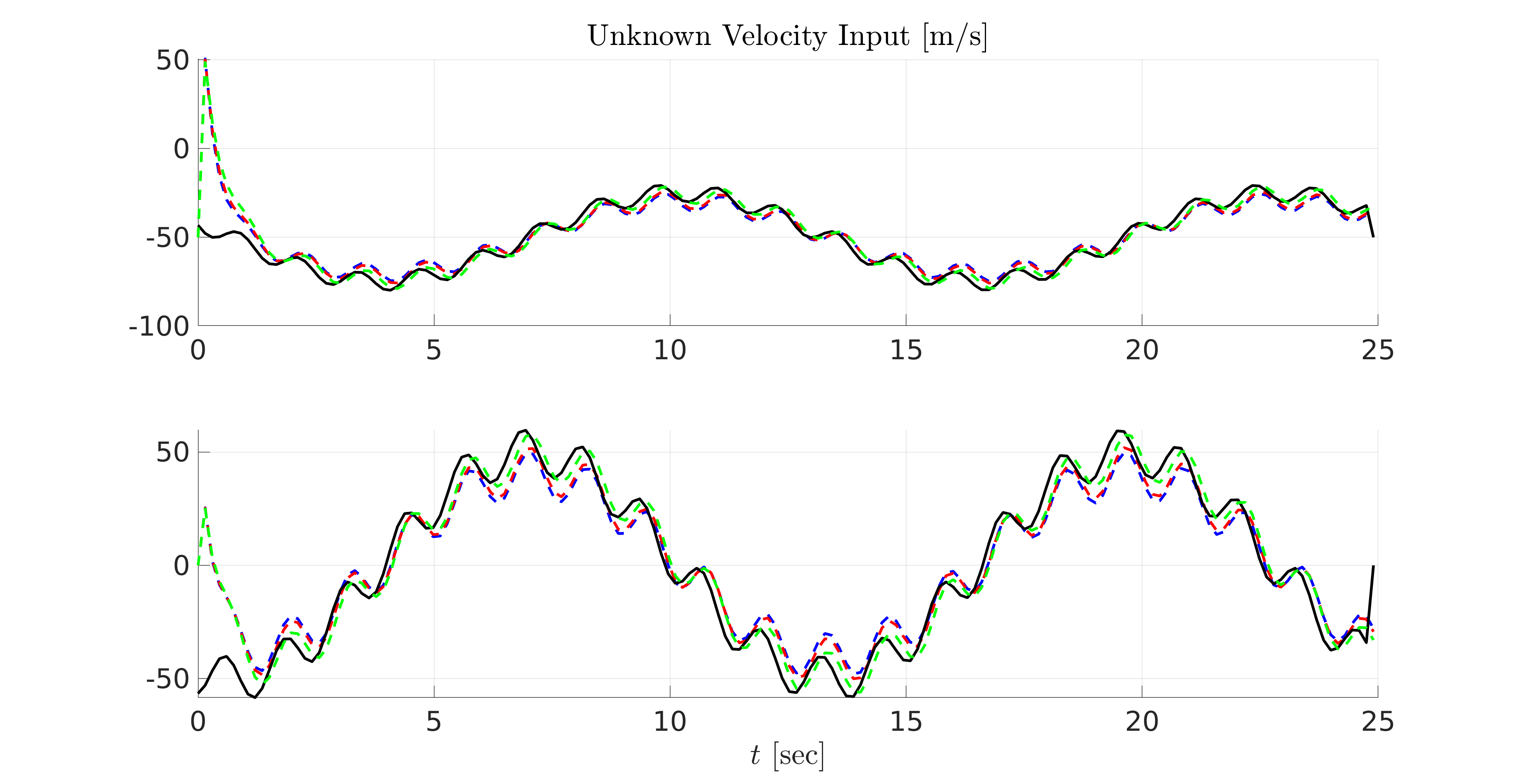}
\end{minipage}
\begin{minipage}{0.49\textwidth}
\hspace{-0.5cm}
\includegraphics[width=1.15\textwidth]{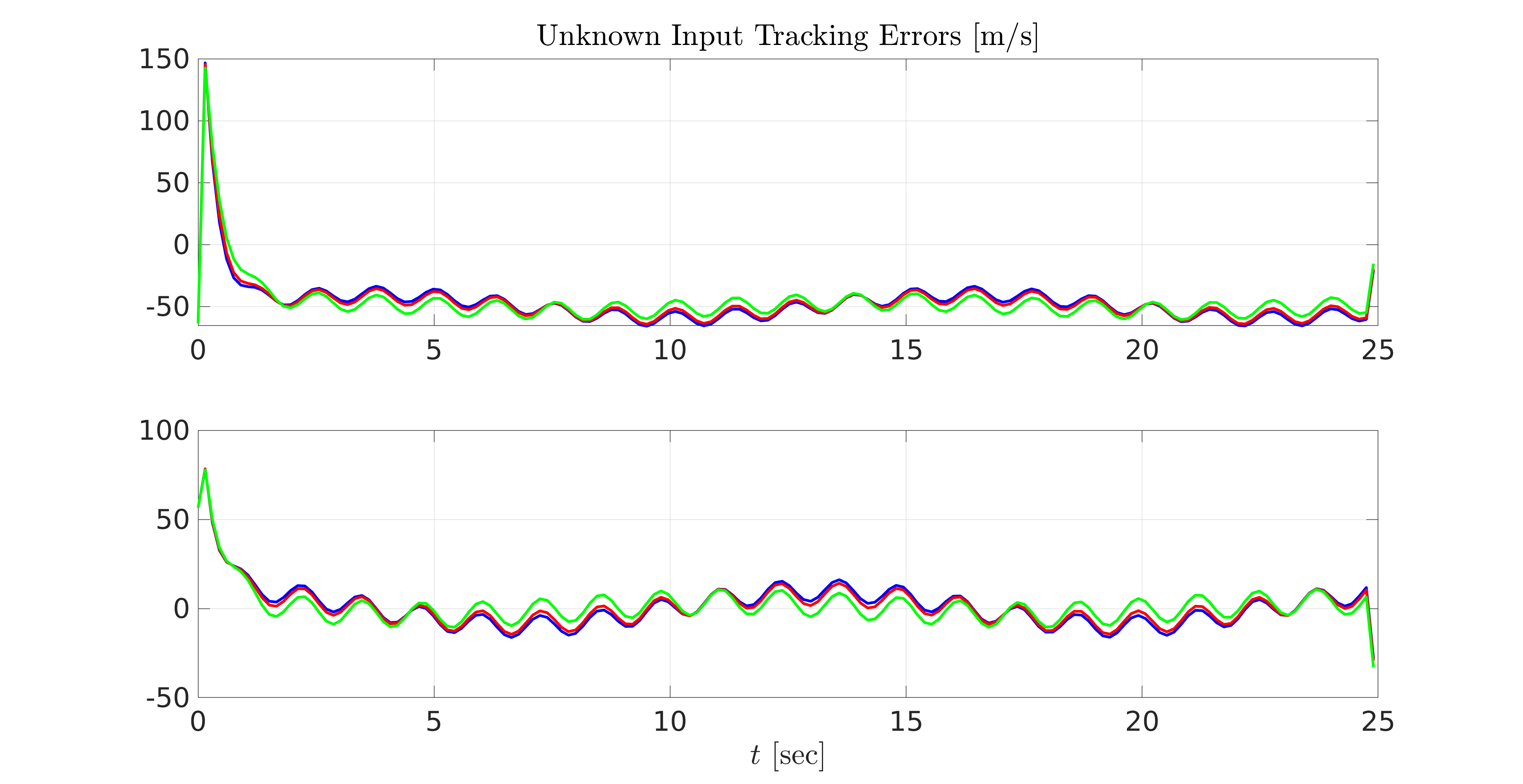}
\end{minipage}
\vspace{-0.1cm}
\caption {Unknown input tracking error results with different values of the design parameter $\alpha$ and the desired performance level satisfying $\mu=\sqrt{\gamma}\leq 0.25$. The curve colors are associated with the design parameters $\alpha_1=0.5$ (blue), $\alpha_2=0.1$ (red), and $\alpha_3=0.01$ (green), respectively. The black curves in the right plot belong to the actual unknown input acting on UAV\#1.}\label{fig:simResultsLocation3}
\vspace{-0.3cm}
\end{figure*}


The proposed UIO with its prediction step in~\eqref{eq:observer} and its unknown input estimation step in~\eqref{eq:inputEstimate} provides estimates of both the next position of the UAVs and the estimated unknown velocity input vector $\mathbf{W}_k$ that is driving the dynamics of the UAV network.  In order to find the UIO state matrix $\mathbf{Q}$ and the UIO gain matrix $\mathbf{L}$, we need to solve the LMIs given by~\eqref{eq:LMI}.  

\begin{figure*}[!t] 
\centering
\includegraphics[width=0.99\textwidth]{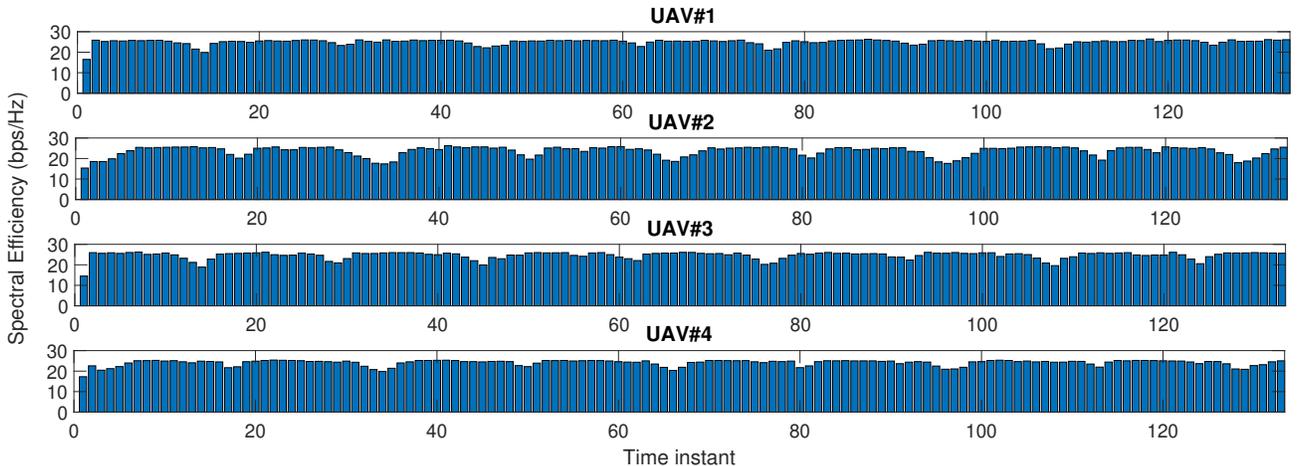}
\vspace{-0.2cm}
\caption{Spectral efficiency of the proposed method for each UAV at each time instant. }\label{fig:spectralEff_all}
\vspace{-0.25cm}
\end{figure*}

\begin{figure*}[!t] 
\centering
\includegraphics[width=0.99\textwidth]{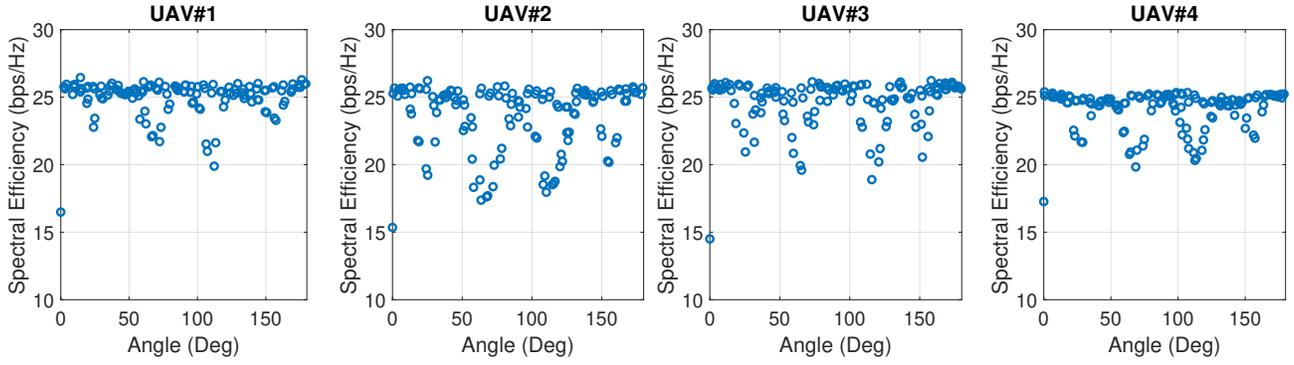}
\vspace{-0.1cm}
\caption{Spectral efficiency of each UAV plotted against different Azimuth angles. }\label{fig:spectEff_theta}
\vspace{-0.15cm}
\end{figure*}

\begin{figure*}
\centering
\begin{tabular}{ccc}
\includegraphics[width=0.67\columnwidth]{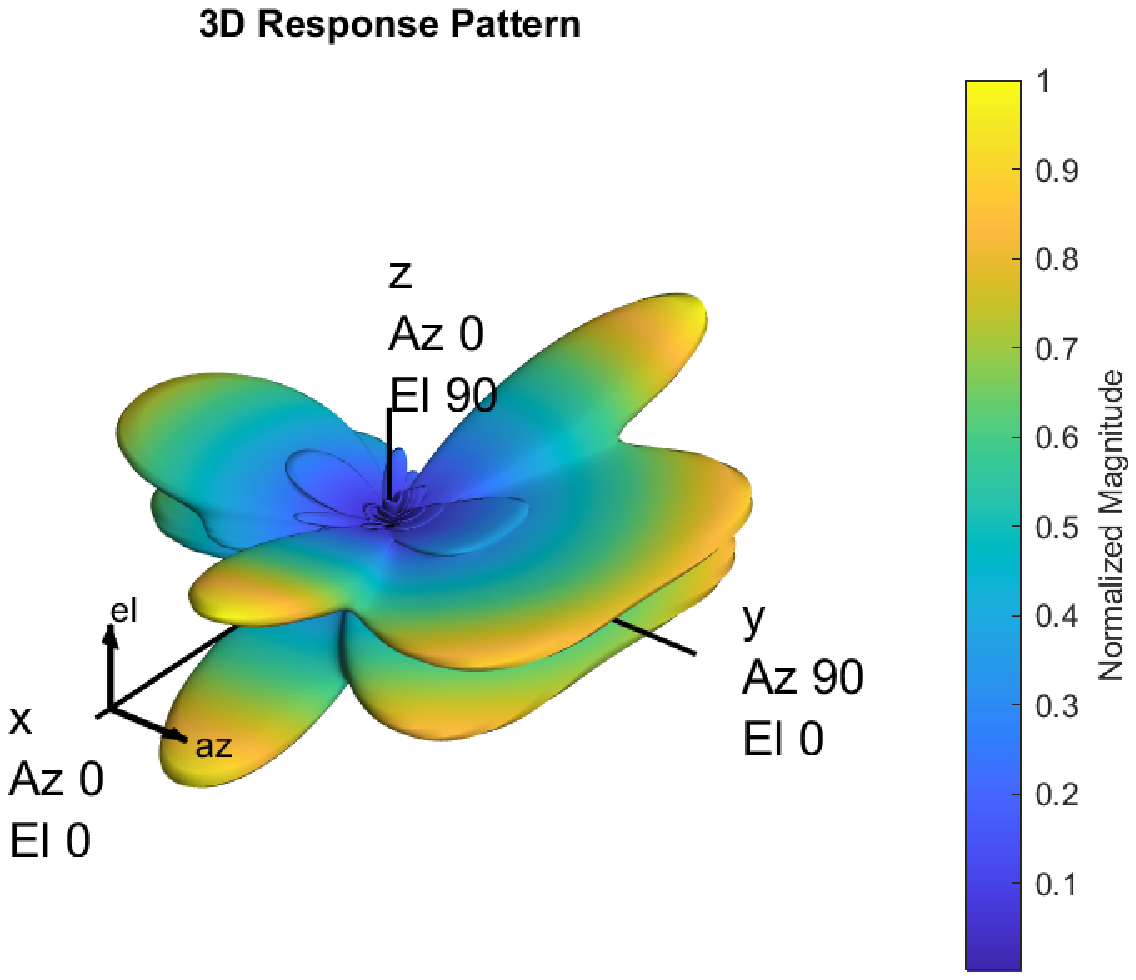} &
\hspace{-0.4cm}
\includegraphics[width=0.67\columnwidth]{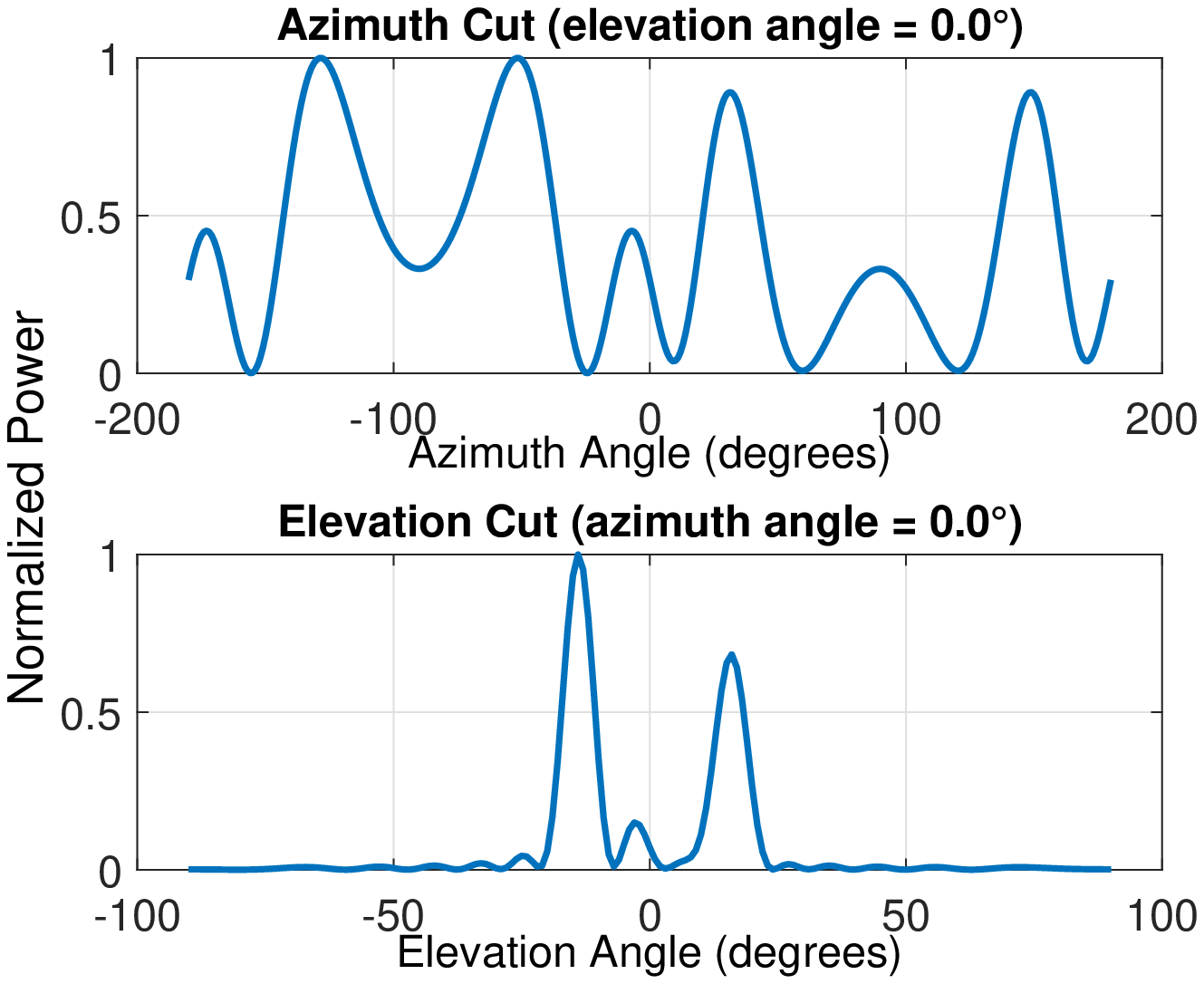}
\hspace{-0.2cm}
\includegraphics[width=0.67\columnwidth]{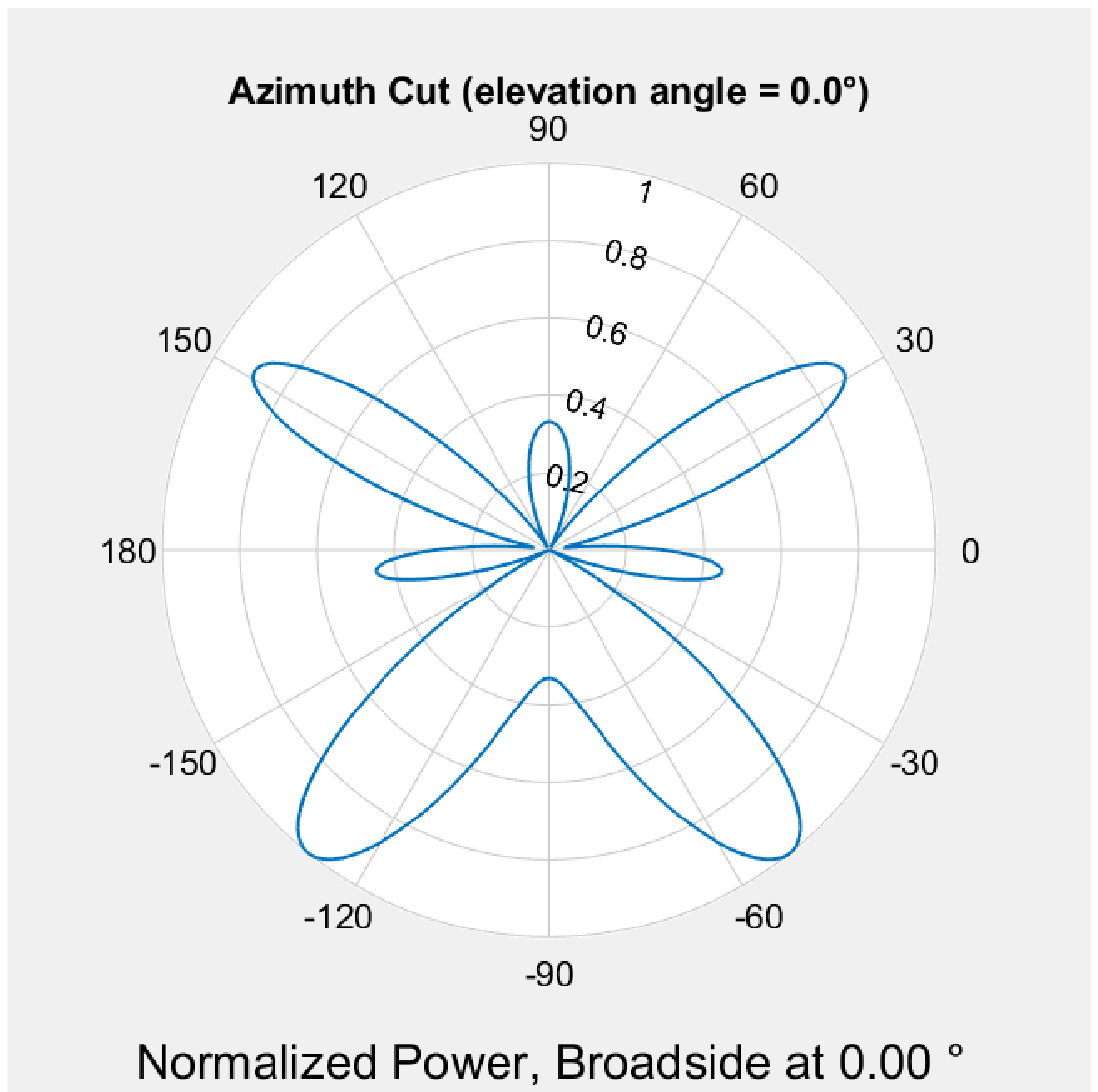}
 \\
\hspace{-0.1cm} ({a})   & \hspace{-0.45cm} ({b}) \hspace{5cm}({c})
\end{tabular}
\vspace{-0.2cm}
\caption{{{(a)~The 3D response pattern of the beamformer, (b)~the normalized power of the beamformer for Azimuth and Elevation cuts, and (c)~the polar pattern demonstrating the directionality of the beamformer when there are $M_{CE}=64$ transmit antennas on the central UAV and $N_U=4$ antennas on each one of the other UAVs.}  
}}\label{fig:patterns64_4}
\end{figure*}

\begin{figure*}
\centering
\begin{tabular}{ccc}
\includegraphics[width=0.67\columnwidth]{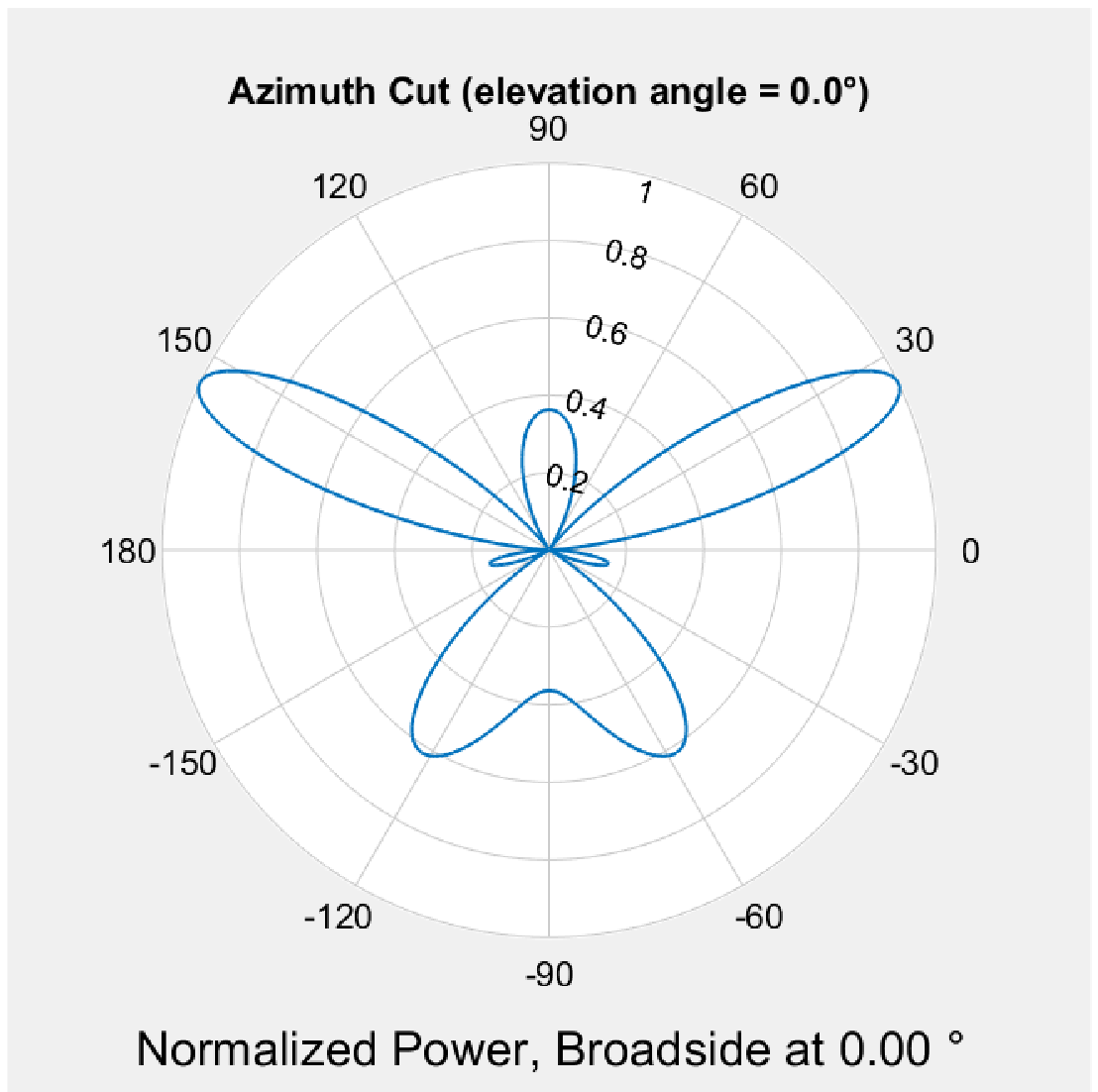} &
\hspace{-0.4cm}
\includegraphics[width=0.67\columnwidth]{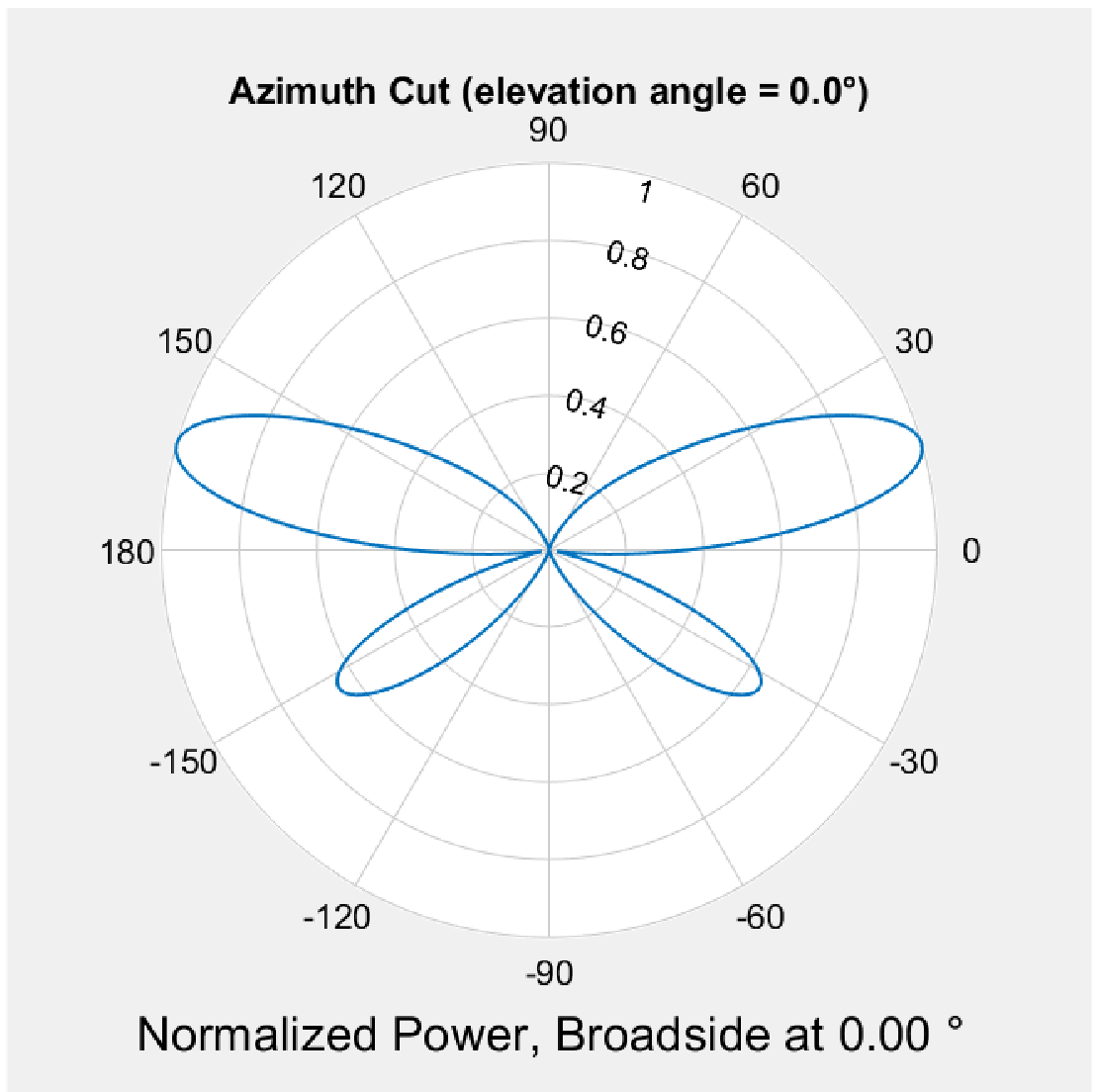}
\hspace{-0.2cm}
\includegraphics[width=0.67\columnwidth]{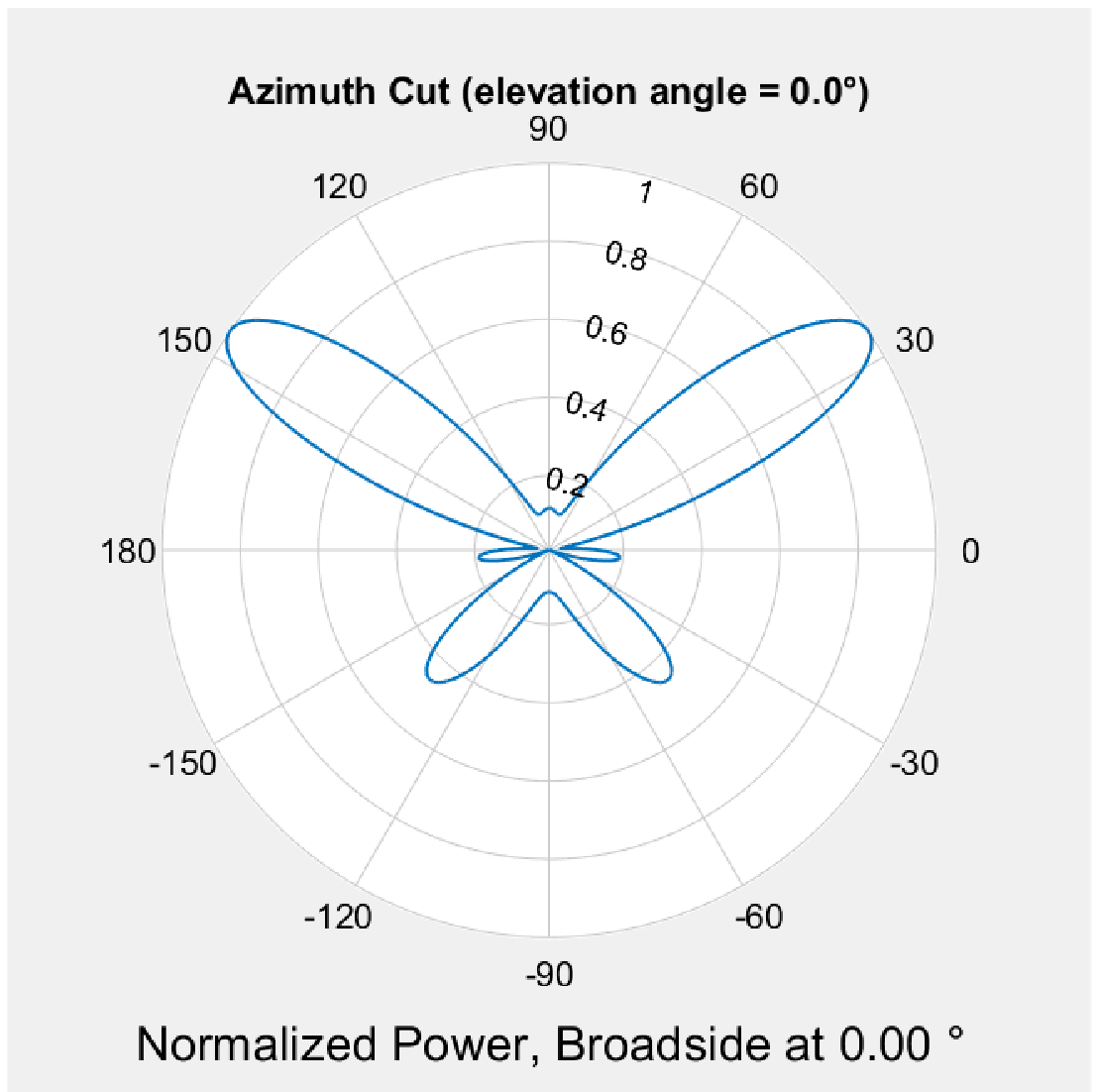}
 \\
\hspace{-0.1cm} ({a})   & \hspace{-0.45cm} ({b}) \hspace{5cm}({c})
\end{tabular}
\vspace{-0.2cm}
\caption{{{As the UAVs fly around the central UAV, the beamformer dynamically alters the beam towards the desired targets according to the angular position estimates provided by the UIO, as illustrated in (a) and (b) for two different angular locations of the UAVs. The impact of increasing the number of antennas is shown in (c) on the same beamformer as Figure~\ref{fig:patterns64_4}, but with $M_{CE}=128$ transmit antennas on the central UAV.}        
}}\label{fig:patternpolar}
\vspace{-2mm}
\end{figure*}

\begin{figure*}
\centering
\begin{tabular}{ccc}
\includegraphics[width=0.67\columnwidth]{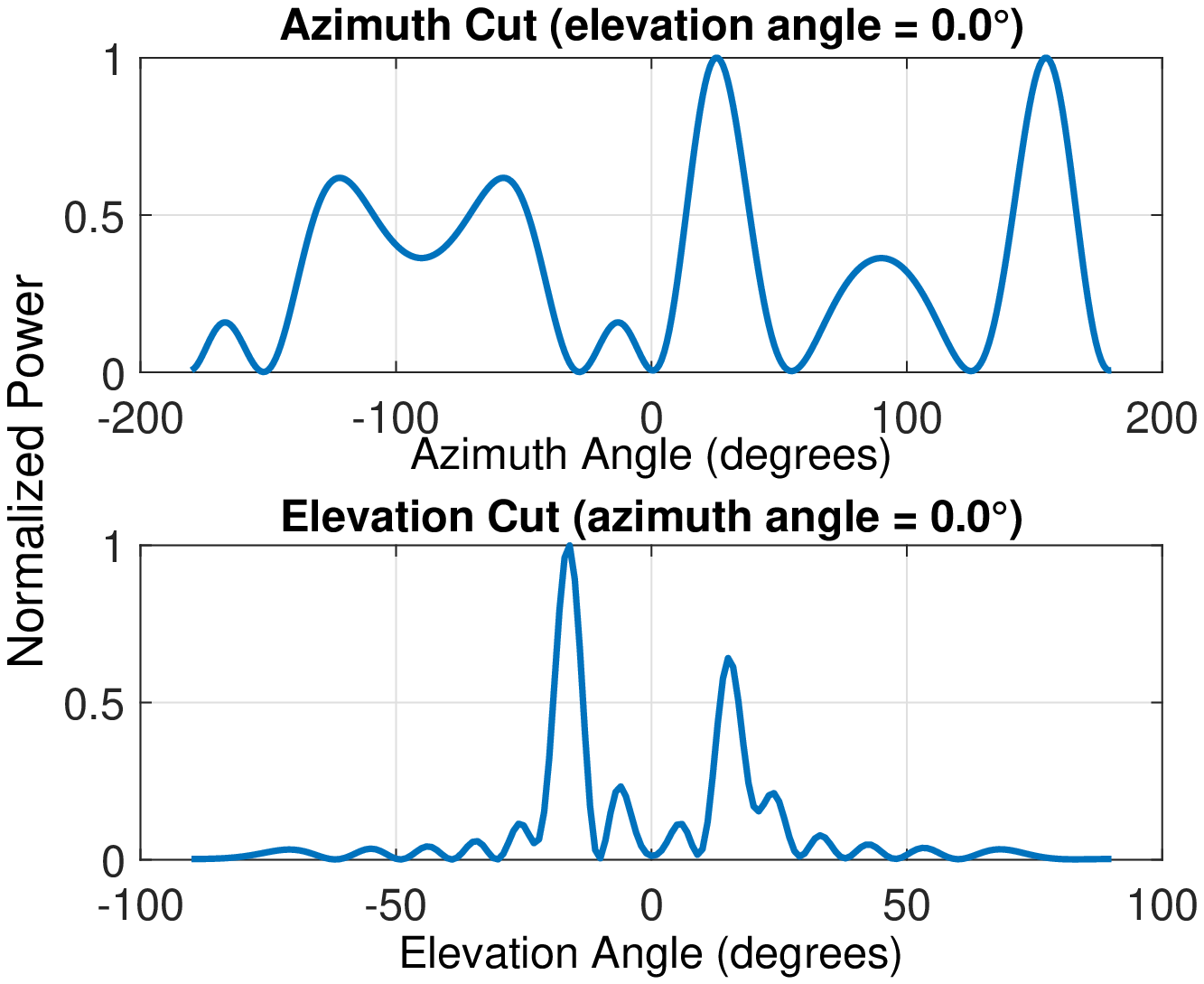} &
\hspace{-0.4cm}
\includegraphics[width=0.67\columnwidth]{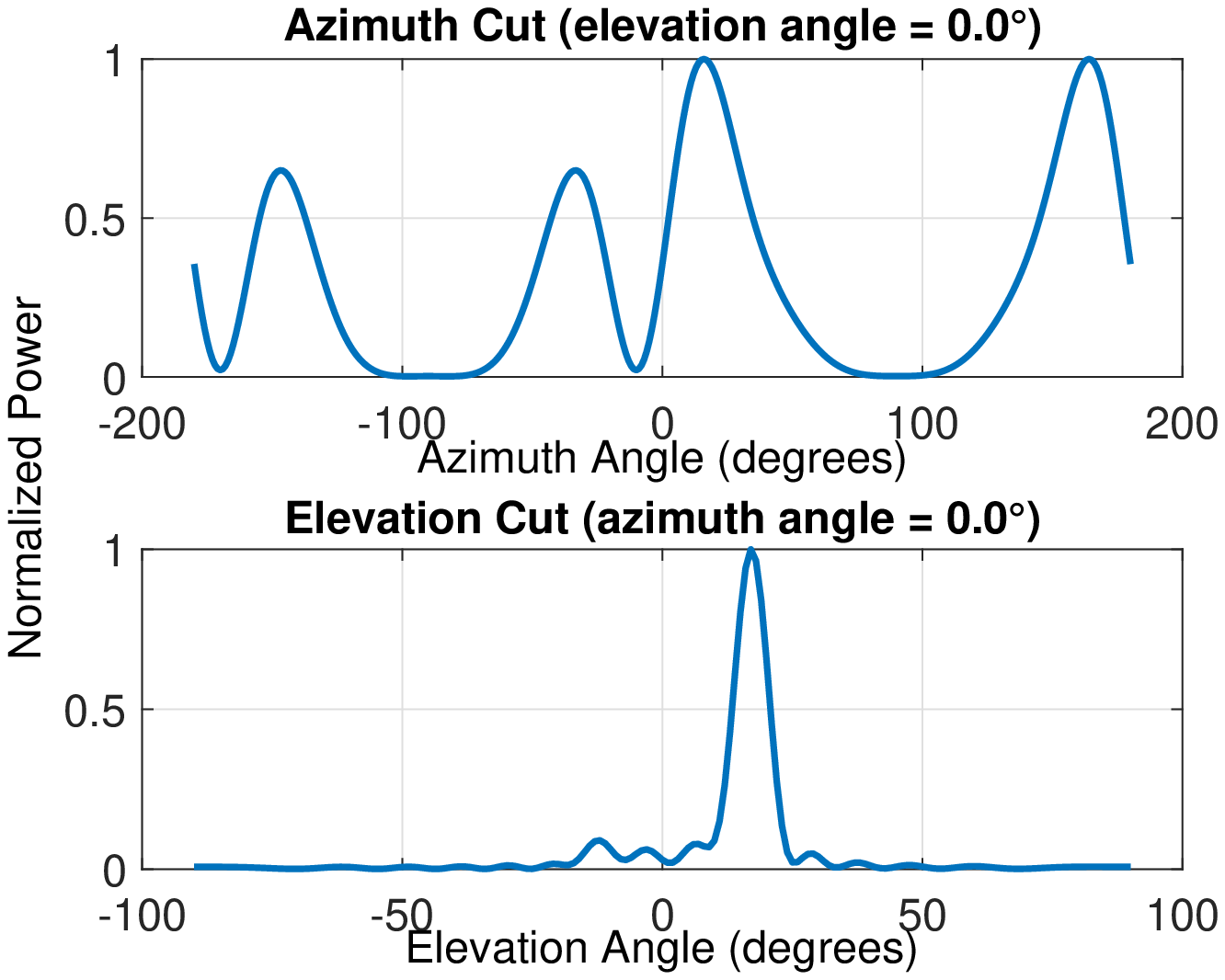}
\hspace{-0.2cm}
\includegraphics[width=0.67\columnwidth]{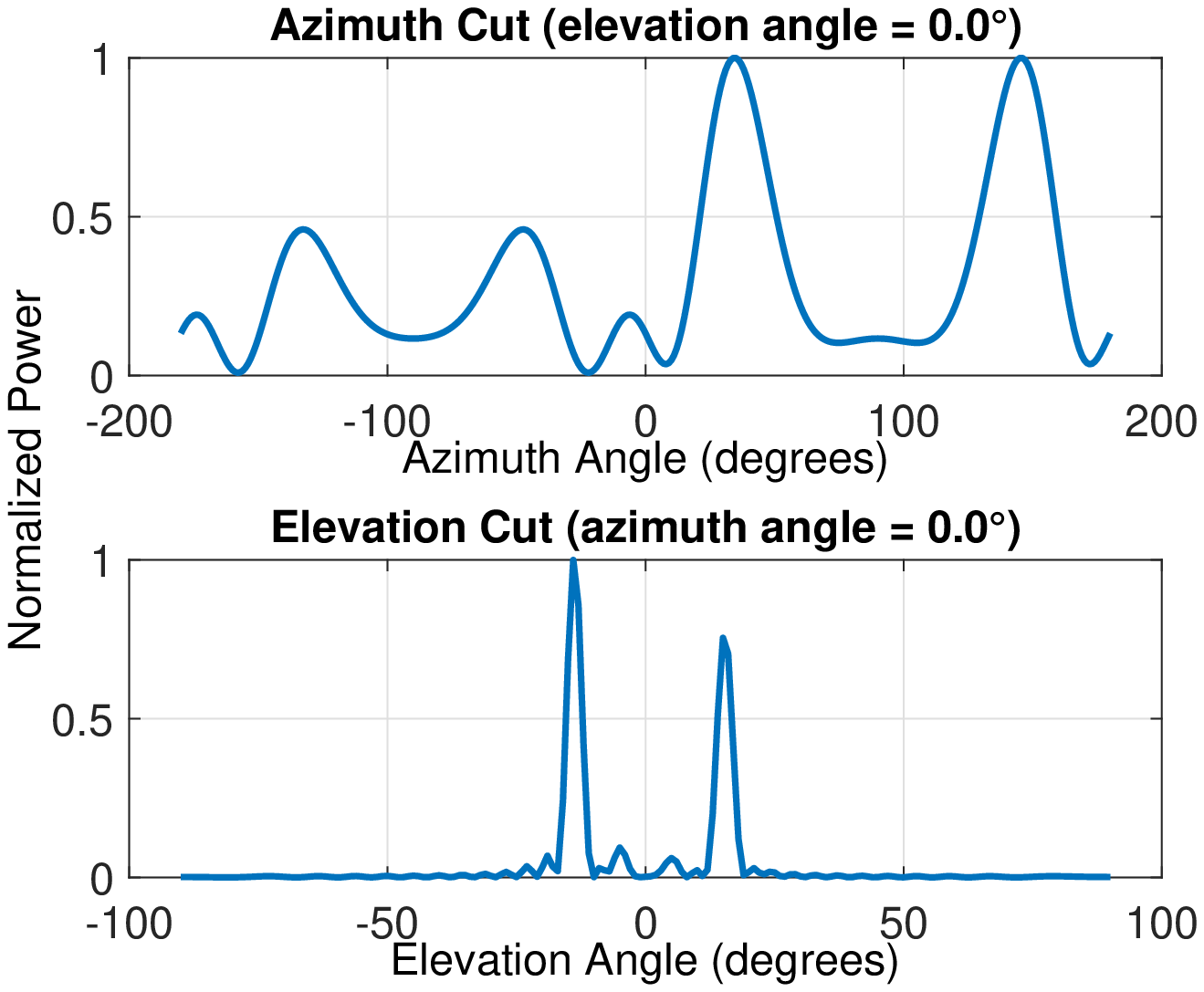}
 \\
\hspace{-0.1cm} ({a})   & \hspace{-0.45cm} ({b}) \hspace{5cm}({c})
\end{tabular}
\vspace{-0.2cm}
\caption{{{The normalized power of the beamformers corresponding to the results in Figure~\ref{fig:patternpolar}.}        
}}\label{fig:patterncut}
\end{figure*}

\begin{figure}[!t] 
\centering
\includegraphics[width=0.94\columnwidth]{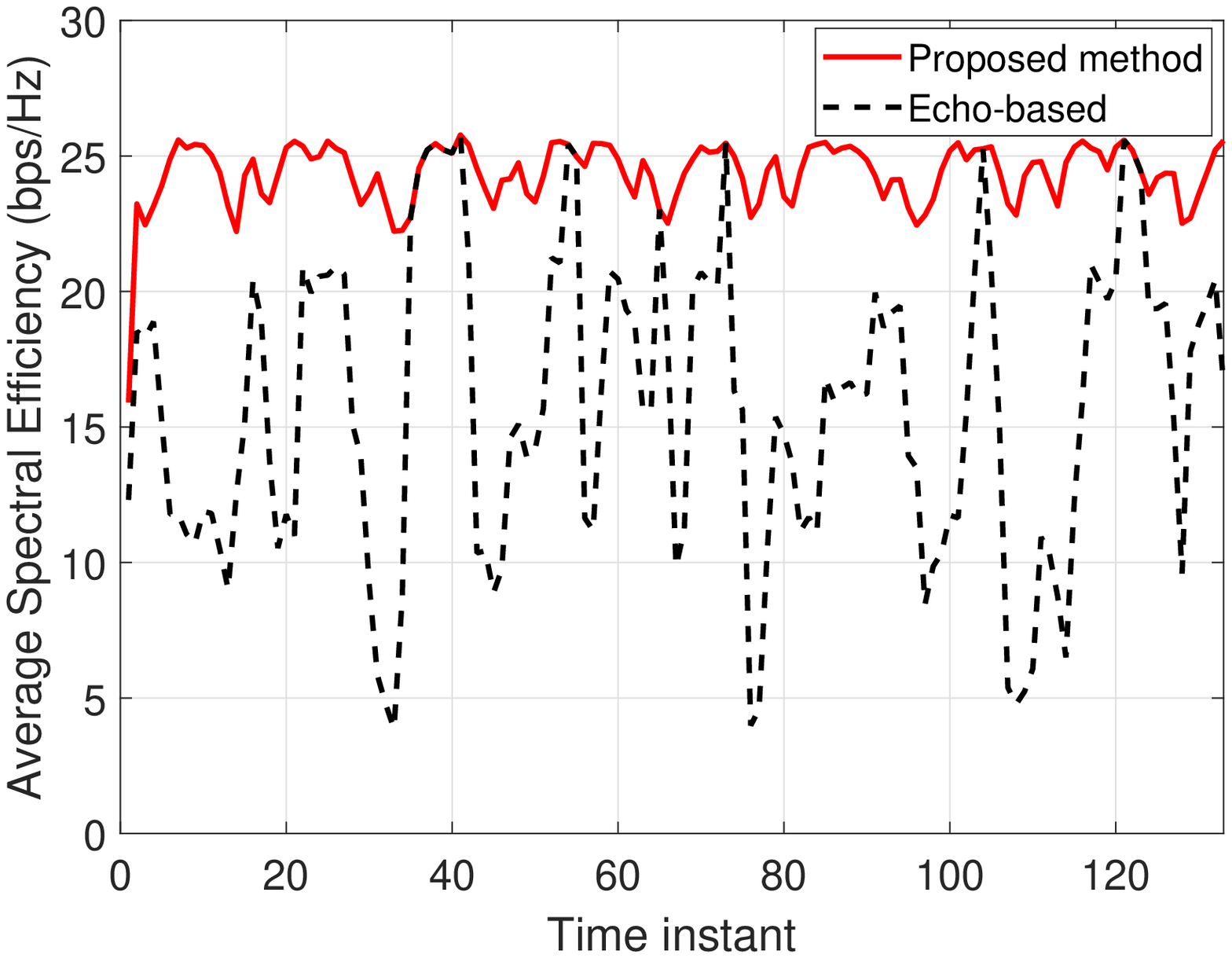}
\vspace{-0.2cm}
\caption{Spectral efficiency of the proposed method as compared with the conventional scenario with no location prediction in which  the probing signal is utilized to estimate the echo for sensing. }\label{fig:spectralEff__ave_all}
\vspace{-0.25cm}
\end{figure}

We solved the LMIs in our simulation studies using the MATLAB CVX package~\cite{grant2020cvx}. According to Proposition~1, the specific LMIs that arise in the context of our UAV location prediction problem will necessarily admit a solution provided that the upper bounds on the performance level $\gamma$ are not too tight with respect to the given measurement time intervals $\Delta T_i$. To solve the LMIs given by~\eqref{eq:LMI}, we have chosen a value of $0.5$ for the constant $\alpha$ in the LMIs. Moreover, we have chosen three different upper bounds on the performance levels for the design parameter $\gamma$.  In particular, we solved the LMIs in~\eqref{eq:LMI} under the three performance level constraints of $\mu_1\leq 0.05$, $\mu_2\leq 0.25$, and $\mu_3\leq 1$, where $\gamma_i = \sqrt{\mu_i}$, $1 \leq i \leq 3$. The obtained UIO state and gain matrices in these three different scenarios are $\mathbf{L}_1=0.39 \mathbf{I}_8$, $\mathbf{Q}_1= 0.61 \mathbf{I}_8$, $\mathbf{L}_2=0.60 \mathbf{I}_8$, $\mathbf{Q}_2=0.40\mathbf{I}_8$, $\mathbf{L}_3=0.76 \mathbf{I}_8$, and $\mathbf{Q}_3=0.24 \mathbf{I}_8$, respectively. Furthermore, the resulting performance levels associated with the found solutions are equal to $\gamma_1=0.21$, $\gamma_2=0.47$, and $\gamma_3=0.96$, respectively (also, see ``\textbf{The  effect  of  increasing  measurement  time  intervals}''). In another simulation scenario, we solved the LMIs in ~\eqref{eq:LMI} under the the same performance level constraint of $\mu\leq 0.25$ and three different design parameters $\alpha$ in~\eqref{eq:LMI}. It can be seen that by decreasing the parameter $\alpha$ the tracking errors become smaller. 

Figure~\ref{fig:simResultsLocation} depicts the actual and predicted location of the UAVs in the network. Figure~\ref{fig:simResultsLocation2} depicts the position and unknown input tracking error profiles resulting from our proposed UIO-based framework. By Proposition~\ref{prop:state}, the steady state position tracking errors are guaranteed to be bounded by $\gamma \| \mathbf{W}\|_\infty$, where $\gamma$ is the performance level of the UIO and $\| \mathbf{W}\|_\infty$ is the maximum perturbation amplitude. Similarly, by Proposition~\ref{prop:input}, the steady state unknown input tracking errors are guaranteed to be bounded by $3 \gamma \| \mathbf{W}\|_\infty$. Figure~\ref{fig:simResultsLocation3} depicts unknown input tracking error results with the desired performance level satisfying $\mu=\sqrt{\gamma}\leq 0.25$ and different values of the design parameter given by $\alpha_1=0.5$, $\alpha_2=0.1$, and $\alpha_3=0.01$, respectively. 

Figure~\ref{fig:spectralEff_all} displays the time profile of spectral efficiency of each UAV when the  beamformer in~\eqref{eq:beamformer} is fed with the angular position estimates that are generated by the UIO according to~\eqref{eq:angPos}. Thanks to the reliable sensing information provided by the proposed observer, the spectral efficiency is rather resilient as shown in Figure~\ref{fig:spectralEff_all}. The spectral efficiency of the UAVs in space when they are at different angular positions with respect to  the central UAV are depicted in Figure~\ref{fig:spectEff_theta}. 

For $M_{CE}=64$ transmit antennas on the central UAV and $N_U=4$ antennas on the other $4$ UAVs, the 3D response pattern of the beamformer, the normalized power of the beamformer for Azimuth and Elevation cuts, and the polar pattern for displaying the directionality of the beamformer are presented in Figure~\ref{fig:patterns64_4}.
Since the UAVs fly around the central UAV, their angular positions vary with time. Accordingly, the beamformer dynamically modifies the beam towards the intended targets. Figures~\ref{fig:patternpolar}(a)-(b) depict the resulting beams in two different time instants.
Figure~\ref{fig:patternpolar}(c) depicts the effect of the number of antennas on the same beamformer as previously plotted in Figure~\ref{fig:patterns64_4}, but with $M_{CE}=128$ transmit antennas on the central UAV instead of $64$ as assumed in Figure~\ref{fig:patterns64_4}.
Figure~\ref{fig:patterncut} shows the normalized power of the beamformers in azimuth and elevation planes corresponding to the angles in Figure~\ref{fig:patternpolar}.
Figure~\ref{fig:spectralEff__ave_all} shows the average spectral efficiency of the proposed solution in comparison with a conventional scenario in which a probing signal is sent by the central UAV. The echo wave reflected by the target UAVs is received by the central UAV if
the blockage does not occur. In the case of blockage, as simulated at some specific time instants, the average spectral efficiency in the case of conventional echo-based scenario drops significantly,  whereas the proposed UIO-based technique exhibits high resistance to this phenomena of spectral efficiency drop. 
 
\noindent\textbf{The effect of increasing measurement time intervals.} It is expected that with increasing the amount of interval delays $\Delta T_i$ in~\eqref{eq:singleDrone} for updating the position of the UAVs, the prediction of their positions becomes more challenging. Indeed, this has been the case in our simulations. In particular, in another round of simulations, we started increasing the time interval delays for the three different performance level constraints  $\mu_1\leq 0.05$, $\mu_2\leq 0.25$, and $\mu_3\leq 1$ up to the point where the LMI system in~\eqref{eq:LMI} failed to have a solution. For these three different performance level constraints, the LMI system in~\eqref{eq:LMI} ceased to admit a solution at $\Delta T=0.89$ sec, $\Delta T=1.01$ sec, and $\Delta T=1.3$ sec, respectively. 
 
\section{Conclusion} \label{sec:conc}
In this paper, we proposed a novel spatial estimation algorithm for beamforming in UAV-based JSC. The proposed algorithm, which does not make any \emph{a priori} assumptions on the statistics of the UAV velocities and unknown inputs, relies on using UIOs that are designed by solving LMIs. Our observer-based location prediction algorithm  reduces the joint target detection and communication complication notably  by operating on the same device and performs reliably in the presence of channel blockage and interference.


%
%
%
%



\ifCLASSOPTIONcaptionsoff
  \newpage
\fi

\bibliographystyle{IEEEtran}
\bibliography{ref}

\end{document}